\documentclass[acmsmall,nonacm]{acmart}
\allowdisplaybreaks
\usepackage[ruled,vlined,linesnumbered]{algorithm2e}
\usepackage[utf8]{inputenc}
\usepackage{graphicx,enumitem,booktabs,amsmath}
\usepackage{tikzscale}
\usepackage{tikz}
\usetikzlibrary{shapes.geometric, positioning, fit, backgrounds, calc, decorations.pathmorphing, decorations.pathreplacing, arrows.meta, patterns}
 \usepackage{wrapfig}

\usepackage{mathabx}
\usepackage{lineno}

\usepackage[para,online,flushleft]{threeparttable}
\usepackage{psfrag}
\usepackage{subcaption}
\usepackage{tikz}
\usetikzlibrary{arrows}
\usepackage{marginnote}
\usepackage{color}

\usepackage{multirow}

\newcommand{\E}{\mathcal{E}}
\newcommand{\V}{\mathcal{V}}

\newcommand{\T}{\mathcal{T}}

\renewcommand{\H}{\mathcal{H}}

\newcommand{\rwa}[1]{{\color{black}#1}}
\newcommand{\rwb}[1]{{\color{black}#1}}

\newcommand{\rwf}[1]{{\color{black}#1}}

\newcommand{\dom}{\mathrm{dom}}  
\renewcommand{\path}{\mathcal{Q}_{\textup{\textsf{path}}}}

\newcommand{\N}{\mathcal{N}}

\newcommand{\cut}[1]{}

\newcommand{\cupdot}{\mathbin{\mathaccent\cdot\cup}}

\setcopyright{cc}
\setcctype{by-nc-nd}
\acmJournal{PACMMOD}
\acmYear{2026} \acmVolume{4} \acmNumber{2 (PODS)} \acmArticle{121}
\acmMonth{5} \acmDOI{10.1145/3801917}

\received{December 2025}
\received[accepted]{February 2025}

\author{Qichen Wang}
\affiliation{%
  \institution{Nanyang Technological University}
  \country{Singapore}
}
\email{qichen.wang@ntu.edu.sg}
\orcid{0000-0002-0959-5536}

\begin{document}
\title{Towards Parameterized Hardness on Maintaining Conjunctive Queries}

\begin{abstract}
We investigate the fine-grained complexity of dynamically maintaining the result of fixed self-join free conjunctive queries under single-tuple updates. Prior work shows that free-connex queries can be maintained in update time $O(|D|^{\delta})$ for some $\delta \in [0.5, 1]$, where $|D|$ is the size of the current database. However, a gap remains between the best known upper bound of $O(|D|)$ and lower bounds of $\Omega(|D|^{0.5-\epsilon})$ for any $\epsilon \ge 0$.

We narrow this gap by introducing two structural parameters to quantify the dynamic complexity of a conjunctive query: the height $k$ and the dimension $d$. We establish new fine-grained lower bounds showing that any algorithm maintaining a query with these parameters must incur update time $\Omega(|D|^{1-1/\max(k,d)-\epsilon})$, unless widely believed conjectures fail. These yield the first super-$\sqrt{|D|}$ lower bounds for maintaining free-connex queries, and suggest the tightness of current algorithms when considering arbitrarily large $k$ and~$d$.

Complementing our lower bounds, we identify a data-dependent parameter, the generalized $H$-index $h(D)$, which is upper bounded by $|D|^{1/d}$, and design an efficient algorithm for maintaining star queries, a common class of height 2 free-connex queries. The algorithm achieves an instance-specific update time $O(h(D)^{d-1})$ with linear space $O(|D|)$. This matches our parameterized lower bound and provides instance-specific performance in favorable cases.
\end{abstract}



\maketitle
\sloppy
\section{Introduction}
Answering conjunctive queries under updates is a fundamental task in database theory and systems, with applications to incremental view maintenance~\cite{chirkova2012materialized, wang2023change, nikolic2018incremental, idris17:_dynam, wang2020maintaining, kara2020maintaining}, complex event processing \cite{PintoR24}, and streaming processing~\cite{carbone15:_apach_flink}. \rwa{In the fully dynamic setting, a fixed self-join free conjunctive query\footnote{In the remainder of this paper, we consider only ``self-join free" conjunctive queries and will omit ``self-join free" from the text.  A conjunctive query with self-joins may be easier than the associated self-join free conjunctive query, where the associated self-join free query is obtained by giving distinct relational symbols to all relations \cite{10.1145/3584372.3588667}.  While some lower bounds have been established on queries with self-joins (e.g., Theorem~1.2 in \cite{berkholz17:_answer} on maintaining Boolean queries with self-joins), the hardness results for queries with self-joins are not currently settled in general for static and dynamic evaluations.} 
$Q$ is given together with an initially empty database $D$. The database then undergoes a sequence of tuple insertions and deletions. The goal is to maintain a data structure so that the current query result $Q(D)$ can be retrieved at any time. In particular, let $\omega(Q)$ be the \emph{optimal} static exponent, meaning that for every database instance $D$, all algorithms that can evaluate query $Q$ over $D$ requires runtime $\Omega(|D|^{\omega(Q)-\epsilon}+|Q(D)|)$ for some constant $\epsilon \ge 0$. We are interested in the dynamic algorithms where, after each update, the maintained data structure supports retrieving $Q(D)$ with strictly smaller overhead than re-evaluating $Q$ from scratch, i.e., in time $O(|D|^{\omega(Q)-\epsilon}+|Q(D)|)$ for some constant $\epsilon > 0$. Here, $|D|$ denotes the size of the current database, and $|Q(D)|$ denotes the size of the current query result.  We measure efficiency by the amortized update time $t_u$ for processing a single insertion or deletion, as well as the enumeration delay.  Ideally, the enumeration delay should be constant, i.e., independent of $|D|$.}

Despite this restriction, the dynamic problem remains challenging. Recent works have established nontrivial upper and lower bounds. \rwa{The optimal static exponent $\omega(Q)$ provides a straightforward lower bound $\Omega(|D|^{\omega(Q)-1-\epsilon})$ on the amortized update time \cite{10.1145/3695837,pods25} via an insertion-only reduction. Specifically, if we construct a database instance $D$ by inserting tuples one by one, a dynamic algorithm with amortized update time $t_u$ and enumeration time $O(|D|^{\omega(Q)-\epsilon}+|Q(D)|)$ would allow evaluating $Q$ in $O(|D|\cdot t_u+|D|^{\omega(Q)-\epsilon}+|Q(D)|)$. Consequently, for any $\gamma > 0$, we must have $t_u=\Omega(|D|^{\omega(Q)-1-\gamma})$ even under insertion-only updates; otherwise, the optimality of $\omega(Q)$ would be violated.
Although the lower bound is tight \cite{10.1145/3695837, pods25} for insertion-only update sequences on free-connex queries (i.e., acyclic queries with $\omega(Q) = 1$), it is not the best lower bound when considering deletions.} \citet{berkholz17:_answer} introduced the notion of q-hierarchical queries, a subset of free-connex queries, and proved a tight dichotomy: assuming the Online Matrix-Vector multiplication (OMv) conjecture, every conjunctive query that is not q-hierarchical cannot be maintained with constant update time while supporting constant delay enumeration. In particular, if a conjunctive query is not q-hierarchical, any dynamic algorithm must incur amortized \rwa{$\Omega(|D|^{1/2-\epsilon})$ for any $\epsilon > 0$} update time to achieve constant delay enumeration.  On the upper-bound side, multiple efforts have been made to efficiently maintain free-connex conjunctive queries~\cite{idris17:_dynam, nikolic2018incremental, wang2023change}.  They demonstrate that free-connex queries can be maintained in linear time $O(|D|)$ per update, while still supporting constant delay enumeration.  \rwa{\citet{kara2020trade} studies trade-offs between amortized update time $t_u$ and enumeration delay $t_d$ for hierarchical queries, another subset of acyclic queries that intersects free-connex queries.  For free-connex hierarchical queries, they show a product trade-off $t_d \times t_u = O(|D|)$.  On the other hand, since $\omega(Q) = 1$ for all free-connex queries, any non-constant enumeration delay will result in a larger overhead than re-evaluating the query when $|Q(|D|)| = \Omega(|D|)$.}  Together, these results imply that the best achievable exponent $\delta$ for free-connex queries currently lies between $0.5$ and $1$. Whether we can close this gap remains an open problem in this field.  

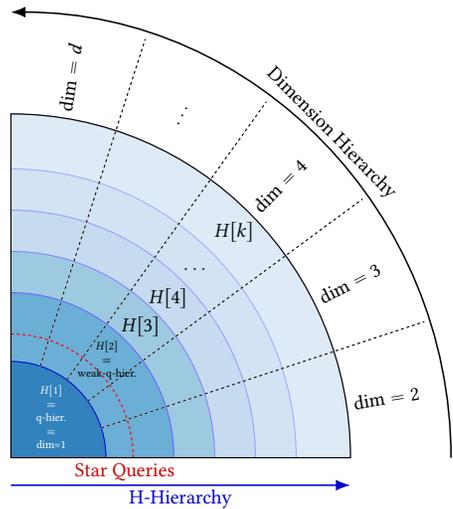
\begin{wrapfigure}{r}{0.45\textwidth}

    \centering
    \vspace{-3em} 
    \resizebox{\linewidth}{!}{
\definecolor{cqCol}{RGB}{245, 245, 245} 
\definecolor{aqCol}{RGB}{230, 230, 235} 

\definecolor{h1Col}{RGB}{49, 130, 189}  
\definecolor{h2InnerCol}{RGB}{107, 174, 214} 
\definecolor{h3Col}{RGB}{158, 202, 225}
\definecolor{h4Col}{RGB}{198, 219, 239}
\definecolor{hdotCol}{RGB}{210, 227, 243}
\definecolor{hkCol}{RGB}{222, 235, 247} 

\definecolor{bergeCol}{RGB}{102, 194, 165} 

\begin{tikzpicture}[font=\Huge]

\def\angStart{0}
\def\angEnd{90}

\def\rOne{3.5}    
\def\rBerge{4.5}    
\def\rTwo{6}    
\def\rThree{7.5}
\def\rFour{9}
\def\rDot{10.5}
\def\rK{12.5}       
\def\rAcyclic{14}
\def\rCQ{16}
\def\rQ{16.2}

\begin{scope}
    \path[clip] (0,0) -- (\angStart:\rCQ) arc (\angStart:\angEnd:\rCQ) -- cycle;

    \draw[fill=white, draw=white] (0,0) -- (\angStart:\rQ) arc (\angStart:\angEnd:\rQ) -- cycle;

    \draw[fill=hkCol, draw=blue!30] (0,0) -- (\angStart:\rK) arc (\angStart:\angEnd:\rK) -- cycle;

    \draw[fill=hdotCol, draw=blue!35] (0,0) -- (\angStart:\rDot) arc (\angStart:\angEnd:\rDot) -- cycle;

    \draw[fill=h4Col, draw=blue!40] (0,0) -- (\angStart:\rFour) arc (\angStart:\angEnd:\rFour) -- cycle;

    \draw[fill=h3Col, draw=blue!50] (0,0) -- (\angStart:\rThree) arc (\angStart:\angEnd:\rThree) -- cycle;

    \draw[fill=h2InnerCol, draw=blue!60] (0,0) -- (\angStart:\rTwo) arc (\angStart:\angEnd:\rTwo) -- cycle;

    \draw[red, dashed, very thick] (\angStart:\rBerge) arc (\angStart:\angEnd:\rBerge);

    \draw[fill=h1Col, draw=blue!80!black, thick] (0,0) -- (\angStart:\rOne) arc (\angStart:\angEnd:\rOne) -- cycle;

    \def\degSplitA{18}
    \def\degSplitB{36}
    \def\degSplitC{54}
    \def\degSplitD{72} 

    \foreach \ang in {\degSplitA, \degSplitB, \degSplitC, \degSplitD} {
        \draw[dashed, ultra thick, white] (\ang:\rOne) -- (\ang:\rCQ+1); 
        \draw[dashed, very thick, black] (\ang:\rOne) -- (\ang:\rCQ+1);
    }
    
    \node[font=\Large, text=white, align=center] at (45:\rOne*0.6) {$H[1]$ \\ $\equiv$ \\ q-hier.\\ $\equiv$ \\ dim=1};
    \node[font=\Large, align=center] at (45:\rBerge+0.5) {$H[2]$ \\ $\equiv$ \\ weak-q-hier.};
    
    \node[font=\Huge] at (45:\rThree*0.9) {$H[3]$};
    \node[font=\Huge] at (45:\rFour*0.91) {$H[4]$};
    \node[font=\Huge] at (45:\rDot*0.92) {$\cdots$};
    \node[font=\Huge] at (45:\rK*0.93) {$H[k]$};

    \node[font=\Huge, rotate=9] at (9:\rAcyclic) {$\dim = 2$};
    \node[font=\Huge, rotate=27] at (27:\rAcyclic) {$\dim = 3$};
    \node[font=\Huge, rotate=45] at (45:\rAcyclic) {$\dim = 4$};
    \node[font=\Huge, rotate=63] at (63:\rAcyclic) {$\cdots$};
    \node[font=\Huge, rotate=81] at (81:\rAcyclic) {$\dim = d$};
\end{scope} 


\draw[very thick, black] (0,0) -- (\angStart:\rK) arc (\angStart:\angEnd:\rK) -- cycle;

\node[text=red!80!black, font=\Huge] at (\rBerge-0.3, -0.5) {Star Queries};

\draw[-{Latex[scale=1.5]}, ultra thick, blue!80!black] (0, -1) -- (\rK, -1) node[midway, below] {H-Hierarchy};

\draw[-{Latex[scale=1.5]}, ultra thick, black] (\rQ, 0) arc (0:90:\rQ) node[midway, sloped, above] {Dimension Hierarchy};





\end{tikzpicture}
    }
    \vspace{-2em}
    \caption{Query Hierarchy with $\omega(Q)=1$.}
    \label{fig:hierarchy}
    \vspace{-1em}
\end{wrapfigure}
\subsection{Our contributions}
In this work, we make progress toward this question by introducing parameterized lower bounds that refine the known dichotomy. We define two new structural parameters of a conjunctive query $Q$: the height (denoted $k$) and the dimension (denoted $d$). Intuitively, the height of $Q$ measures the length of the longest chordless path in $Q$'s relational hypergraph and will increase as we "scale out" the query with more relations.  For a free-connex query $Q$, the height of $Q$ is equal to the minimum height of all free-connex join trees for $Q$.  On the other hand, the dimension of $Q$ measures how the query “scales up” by adding more attributes per relation.\footnote{Formal definitions of the height and dimension parameters are given in Sections~\ref{sec:height} and~\ref{sec:dimension}, respectively.} These parameters induce a natural hierarchy within conjunctive queries. See Figure~\ref{fig:hierarchy} for the hierarchy proposed in this paper on free-connex queries with $\omega(Q) = 1$.  We then prove fine-grained lower bounds that depend on $(k,d)$. In particular, we show:

\begin{theorem}[\textbf{height-based lower bound}]\label{thm:height}
Let $Q$ be a height $k$ query with \rwf{$k\ge2$}. Assuming combinatorial $k$-clique conjecture, no dynamic \textbf{combinatorial} algorithm can maintain $Q$ with amortized update time $O(|D|^{(k-1)/k-\epsilon})$ and support constant delay enumeration for any constant~$\epsilon>0$.
\end{theorem}

\begin{theorem}[\textbf{Dimension-based lower bound}]\label{thm:dimension}
Let $Q$ be a conjunctive query of dimension \rwf{$d \ge 2$}. Assuming the OuMv$_k$ conjecture, no dynamic algorithm can maintain $Q$ with amortized update time $O(|D|^{(d-1)/d-\epsilon})$ and support constant delay enumeration for any constant $\epsilon>0$.
\end{theorem}

Theorem~\ref{thm:height} shows that higher-height queries inherently require larger update costs: for height $k$ one cannot beat the exponents $(k-1)/k$. Theorem~\ref{thm:dimension} shows that queries of dimension $d$ similarly cannot be maintained with an exponent below $(d-1)/d$. \rwb{
In addition, they recover the previous dichotomy from \citet{berkholz17:_answer}:

\begin{proposition}
\label{thm:equiv}
    For any conjunctive query $Q$, the following statements are equivalent: (1) $Q$ is q-hierarchical; (2) $Q$ has height $k = 1$; and (3) $Q$ has dimension $d = 1$.
\end{proposition}
}

\rwb{These results also complement the $\omega(Q)-1$ lower bounds when $\omega(Q) \le 2$.  When $\omega(Q) \le 1.5$, the lower bounds from height and dimension dominate and any maintenance algorithm must incur amortized update time $\Omega\left(|D|^{1-1/\max(k, d)-\epsilon}\right)$, while when $1.5 < \omega(Q) \le 2$, any maintenance algorithm must incur amortized update time $\Omega\left(|D|^{\delta-\epsilon}\right)$ with $\delta = 1-1/\max\left(k, d, \frac{1}{2-\omega(Q)}\right)$ and any constant $\epsilon > 0$.}  In particular, these results are the \textbf{first to provide super-square-root lower bounds for maintaining free-connex queries}, implying that a general $O(|D|^{1-\epsilon})$-time maintenance algorithm for all free-connex queries would violate these conjectures, resolving the open problem in the negative unless breakthrough algorithms for the underlying conjectured-hard problems are found, and thus \textbf{prove the optimality of current upper bound algorithms}.

On the algorithmic side, we complement our lower bounds with a new parameterized upper bound for star queries, a subset of height $2$ queries with details given in Definition~\ref{def:stars}. Specifically:
\begin{theorem}[\textbf{Algorithm for Star Queries}]\label{thm:upper}
Let $Q$ be a star query with dimension $d \ge 1$, and $j \in [0, d]$ output attributes in $\bar{x}_1$, $D$ be the current database instance, and $h(D)$ be its {\em generalized $H$-index} (a measure of database densities, see Section~\ref{sec:upper}). There is a dynamic algorithm with linear space $O(|D|)$ that maintains $Q$ with an amortized update time of $O(h(D)^{d-1})$ and supports $O(h(D)^{d-j})$-delay~enumeration. 
\end{theorem}
Theorem~\ref{thm:upper} shows that for those star queries that are widely used in analytical workloads, we can sometimes do better than the general $(d-1)/d$ bound by exploiting the instance-dependent $H$-index parameter. Meanwhile, $h(D)\le |D|^{1/d}$ always holds; any algorithm that can maintain such queries in time $O(h(D)^{d-1-\epsilon})$ for any $\epsilon > 0$ would contradict Theorem~\ref{thm:dimension} and OuMv$_k$~conjecture. 

\subsection{Organization}
The rest of the paper is structured as follows. In Section~\ref{sec:prelim}, we introduce notation and recall the necessary background used in this paper. We also formally define the height and dimension parameters of a query. In Section~\ref{sec:height}, we study the height parameter, formalizing the notion of a height $k$ query and proving Theorem~\ref{thm:height} via a reduction from the Odd/Even Cycle detection problem. In Section~\ref{sec:dimension} we study the dimension parameter and prove Theorem~\ref{thm:dimension} via a reduction from the OuMv$_k$ problem. In Section~\ref{sec:upper} we present our dynamic algorithm for star queries and prove Theorem~\ref{thm:upper}. We conclude in Section~\ref{sec:con} with a discussion of related work and future directions, with proofs and additional examples in the appendix. 

\section{Preliminaries}
\label{sec:prelim}

\subsection{Conjunctive Queries and Their Structural Properties}
\subsubsection{Acyclic Conjunctive Queries}
\paragraph{Conjunctive Queries (CQs)~\cite{abiteboul1995foundations}}
Let a database schema be defined by a set of attributes $\V = \{x_1, \dots, x_m\}$ and a set of $n$ relations $\E = \{R_1, \dots, R_n\}$. Each relation $R_i \in \E$ has a schema $\bar{x}_i \subseteq \V$. A conjunctive query (CQ) is expressed as:
\[
    Q(\bar{y}) \gets R_1(\bar{x}_1), \cdots, R_n(\bar{x}_n),
\]
where $\bar{y} \subseteq \V$ is the set of output attributes.  If $\bar{y} = \V$, the query is a full conjunctive query, and we can omit $\bar{y}$ from the head, denote it as $Q(*)$.  If $\bar{y} = \emptyset$, the query is a Boolean conjunctive query, and we write it as $Q()$.  We also use $(\V, \E, \bar{y})$ to represent a CQ.    

For any attribute $x \in \V$, we let $\E[x] = \{R(\bar{x}) \in \E : x \in \bar{x}\}$ be the set of relations containing $x$.  We also let $\V_\bullet = \{x \in \V: |\E[x]|=1\}$ be the set of unique attributes, \rwf{i.e., attributes that are not join attributes}.  

\paragraph{Hypergraphs and Acyclicity}
The acyclicity of the associated relational hypergraph defines the acyclicity of a CQ $Q$.  

\begin{definition}[Relational Hypergraph]
The relational hypergraph associated with a conjunctive query $Q = (\V, \E, \bar{y})$ is $\H = (\V, \E)$, where $\E$ is the set of hyperedges (relations) and $\V$ is the set of vertices (attributes). Each hyperedge $R_i \in \E$ contains the set of vertices $\bar{x}_i$.  

A \emph{primal graph} $G = (V, E)$ of $\H$ is a graph with $V = \V$, and an edge $(v, u) \in E$ if and only if there exists $R(\bar{x}) \in \E$, such that $v, u \in \bar{x}$.  A \emph{primal graph} over $V \subseteq \V$ is a graph $G(V) = (V, E)$, and an edge $(v, u) \in E$ if and only if there exists $R(\bar{x}) \in \E$, such that $v, u \in V \cap \bar{x}$.
\end{definition} 
There are multiple notions of acyclicity; in this paper, we consider the standard notion of $\alpha$-acyclicity. A CQ $Q$ is $\alpha$-acyclic if and only if it admits a \emph{generalized join tree} (GJT) $\T$ \cite{wang2023change}.  

A GJT is a tree $\T = (\V(\T), \E(\T))$ where each node $v \in \V(\T)$ corresponds to either an \emph{input relation} $R_i(\bar{x}_i) \in \E$, or a \emph{generalized relation} $[\bar{x}]$ with $[\bar{x}] \subseteq \bar{x}_i$ for some $R_i(\bar{x}_i)$. The GJT $\T$ satisfy:
\begin{itemize}[leftmargin=*]
    \item \textbf{Cover:} For every relation $R_i(\bar{x}_i) \in \E$, there exists a node $v \in \V(\T)$ such that $v = R_i$. Furthermore, all leaf nodes of $\T$ must be input relations from $\E$.
    \item \textbf{Connect:} For any attribute $x \in \V$, the set of nodes $\{v \in \V(\T) : x \in v\}$ induces a connected subtree in $\T$.
\end{itemize}
This definition generalizes traditional join trees, which are special cases where all nodes in $\V(\T)$ correspond to input relations in $\E$.

\subsubsection{Updates, Enumeration and Model of Computation.}
We assume an initially empty database. We consider a dynamic setting where the database undergoes a sequence of single-tuple updates. Each update consists of either inserting or deleting a tuple $t \in \dom(\bar{x})$ from a relation $R(\bar{x})$. We adhere to set semantics: inserting a tuple $t$ that already exists in $R$, or deleting a tuple $t$ that is not present in $R$, results in no change to the database instance. While a specific tuple value may be inserted or deleted multiple times throughout the sequence, we treat each operation as a distinct~event.

Updates are processed sequentially; the maintenance procedure for a single update must conclude before the next update is received. Following any update, the system may allow enumeration of the query result. We define the \emph{enumeration delay} as the maximum time elapsed between the start of the enumeration and the output of the first tuple, between the output of any two consecutive tuples, and between the output of the last tuple and the termination of the procedure. We focus on \emph{data complexity}, treating the query size as a constant. 

Our complexity analysis is based on the standard word RAM model of computation under the uniform cost measure. We assume a machine word size of $w = \Omega(\log N)$ bits for an input of size $N$. In this model, fundamental operations, including arithmetic operations, bitwise logic, concatenation, and direct memory access, are executed in $O(1)$ time.


    

\subsubsection{Free-connex CQs}

The definition of acyclicity for conjunctive queries does not consider the influence of output attributes. However, if we also take the output attributes $\bar{y}$ into account, we can refine the notion of acyclicity to define the class of \emph{free-connex} queries. Specifically, an acyclic CQ $Q = (\V, \E, \bar{y})$ is \emph{free-connex} if the derived query $Q' = (\V, \E \cup \{\bar{y}\}, \bar{y})$, obtained by treating the set of output attributes $\bar{y}$ as an additional relation, is also acyclic.  Equivalently, an acyclic CQ $Q$ is free-connex if it admits a free-connex join tree. 

\paragraph{free-connex join trees.} A GJT $\T$ is a \emph{free-connex join tree} if it satisfies the standard \textbf{Cover} and \textbf{Connect} properties, as well as the following three additional properties:
\begin{itemize}[leftmargin=*]
    \item \textbf{Above:} No input relation node $v \in \V(\T)$ is an ancestor of any generalized relation node.
    \item \textbf{Guard:} For any generalized relation node $v \in \V(\T)$, all its children $v_c \in \V(\T)$ satisfy $v \subseteq v_c$.
    \item \textbf{Connex:} There exists a \emph{connex set} of nodes $\V_{\textsf{con}}(\T) \subseteq \V(\T)$ such that: (i) $\V_{\textsf{con}}(\T)$ induces a connected subtree of $\T$ that contains the root of $\T$; (ii) for each node $v \in \V_{\textsf{con}}(\T)$ with parent $v_p$, their shared attributes are all output attributes, i.e., $(v \cap v_p) \subseteq \bar{y}$; and (iii) the union of attributes in the connex set covers all output attributes, i.e., $\bar{y} \subseteq \bigcup_{v \in \V_{\textsf{con}}(\T)} v$.
\end{itemize}

\rwa{The class of free-connex queries attracts significant attention because only free-connex queries can be evaluated in linear time or enumerated in constant delay with linear preprocessing time, assuming the Boolean matrix multiplication conjecture and Hyperclique conjecture.

\begin{conjecture}[Boolean Matrix Multiplication Conjecture (BMM)]
    The multiplication of two $n\times n$ boolean matrices cannot be done in time $O(n^2)$.
\end{conjecture}

\begin{conjecture}[HyperClique Conjecture (HyperClique)]
    For any integer $k \ge 3$, it is not possible to determine the existence of a $k$ hyperclique in a $(k-1)$-uniform hypergraph with $m$ hyperedges in time $O(m)$.
\end{conjecture}

\begin{theorem}[\cite{bagan2007acyclic}]
    A CQ $Q$ can be evaluated in linear time $O(|D|+|Q(D)|)$ if and only if $Q$ is free-connex, assuming BMM and HyperClique.
\end{theorem}

\begin{theorem}[\cite{bagan2007acyclic}]
    A CQ $Q$ can be enumerated with $O(1)$ delay after a linear-time $O(|D|)$ pre-processing step if and only if $Q$ is free-connex, assuming BMM and HyperClique.
\end{theorem}

These theorems also suggest that a CQ $Q$ has $\omega(Q) = 1$ if and only if $Q$ is free-connex, assuming BMM.  Such dichotomies also extend to dynamic settings with insertion-only update sequences.
\begin{theorem}[\cite{wang2023change, idris17:_dynam, pods25}]
    Given any insertion-only update sequence, a CQ $Q$ can be maintained in amortized $O(1)$ time while supporting $O(1)$-delay enumeration if and only if $Q$ is free-connex, assuming BMM and HyperClique.
\end{theorem}
}
\begin{example}
\label{ex:1}
    Consider the free-connex query
    \[
        \begin{aligned}
        Q_1(x_3, x_4, x_9) \gets & R_1(x_1, x_2), R_2(x_2, x_3), R_3(x_3, x_4, x_5, x_6), R_4(x_4, x_7), \\ &R_5(x_7, x_8), R_6(x_8), R_7(x_4, x_5, x_9), R_8(x_6), R_9(x_3, x_6), R_{10}(x_9).  
        \end{aligned}
    \]
    Figure~\ref{fig:hypergraph} shows its relational hypergraph with output attributes marked in solid dots, and Figure~\ref{fig:jointree} shows a possible generalized join tree, where all nodes in red form the connex subtree.  
\end{example}
\begin{figure}[t]
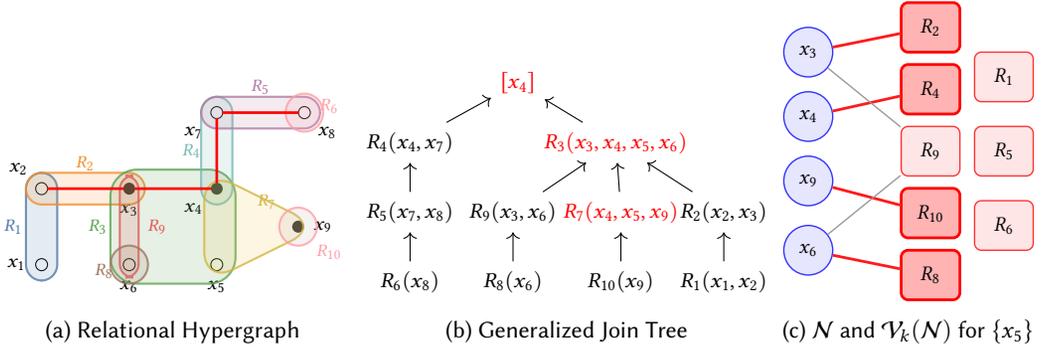

\begin{minipage}[b]{0.34\linewidth}
\includegraphics[width=\linewidth]{hypergraph.tikz}
\subcaption{Relational Hypergraph}
\label{fig:hypergraph}
\end{minipage}
\hfill
\begin{minipage}[b]{0.40\linewidth}
\includegraphics[width=\linewidth]{gentree.tikz}
\subcaption{Generalized Join Tree}
\label{fig:jointree}
\end{minipage}
\hfill
\begin{minipage}[b]{0.24\linewidth}
\includegraphics[width=\linewidth]{bipartite.tikz}
\subcaption{$\mathcal{N}$ and $\V_k(\mathcal{N})$ for $\{x_5\}$}
\label{fig:bipartite}
\end{minipage}
\vspace{-1em}
\caption{Illustration of Example~\ref{ex:1}}
\end{figure}

\subsubsection{Fine-grained Classification of Free-connex Queries}
In the static setting, every free-connex query can be efficiently evaluated and enumerated. However, Berkholz et al.~\cite{berkholz17:_answer} showed that in the dynamic setting, understanding the hardness of maintaining queries requires a more fine-grained classification. In particular, they identified the class of \emph{q-hierarchical queries}, a subset of free-connex queries that can be efficiently maintained.

\paragraph{q-hierarchical CQs}
A CQ $Q = (\V, \E, \bar{y})$ is q-hierarchical if it satisfies the two conditions: (1) for any pair of attributes $x_1, x_2 \in \V$, either $\E[x_1] \subseteq \E[x_2]$ or $\E[x_2] \subseteq \E[x_1]$ or $\E[x_1] \cap \E[x_2] = \emptyset$; and (2) if $x_1 \in \bar{y}$ and $\E[x_1] \subsetneq \E[x_2]$, then $x_2 \in \bar{y}$.  

\begin{theorem}[\cite{berkholz17:_answer},Theorem~1.1]
\label{the:q-hierarchical_lb}
Given any update sequence, a CQ $Q$ can be maintained in $O(1)$ time while supporting $O(1)$-delay enumeration if and only if $Q$ is q-hierarchical, assuming OMv~Conjecture.
\end{theorem}

\paragraph{Weak-q-hierarchical CQs} 
The q-hierarchical class is quite restrictive; for example, even a simple query $Q(x_1) \gets R_1(x_1, x_2), R_2(x_2)$ is not q-hierarchical. Wang et al.~\cite{wang2023change, pods25} slightly extended this class by defining \emph{weak-q-hierarchical} queries, which can be efficiently maintained under certain practical conditions, such as first-in-first-out (FIFO) update sequences. To define weak-q-hierarchical queries, we need to introduce the notions of \emph{ear} and \emph{skeleton}.

The concept of ears was first introduced by the GYO algorithm~\cite{grahamGYO, DBLP:conf/compsac/YuO79}, where a relation $R$ is an ear if there exists another relation $R'$ that contains all non-unique attributes of $R$.  The algorithm tests $\alpha$-acyclicity of a CQ by iteratively removing all unique attributes and relations whose attributes are contained in another relation, or equivalently, iteratively removing ears from the query, and then checking whether the query becomes empty. If the query is reduced to the empty set, one can construct a valid join tree for the query by connecting each removed ear $R$ to its anchor $R'$. 

Hu and Wang~\cite{pods25} extended this notion to account for output attributes, referring to such relations as \emph{reducible relations}. Under this extension, the generalized GYO reduction proceeds as follows: iteratively removing a unique attribute if it is not an output attribute, or if it belongs to a relation $R$ whose non-unique attributes are all output attributes.  One then removes every relation whose attributes are contained in another relation. Based on this generalized reduction, we now give the formal definition of an ear for a free-connex query.

\begin{definition}[Ears]
    In a free-connex query $Q = (\V, \E, \bar{y})$, a relation $R(\bar{x}) \in \E$ is called an ear of the query if there exists another relation $R'(\bar{x}') \in \E$, such that either (i) $\bar{x} \setminus \V_\bullet \subseteq \bar{x}'$ and $\bar{x} \cap \V_\bullet \subseteq \V \setminus \bar{y}$; or (ii) $\bar{x} \setminus \V_\bullet \subseteq \bar{x}' \cap \bar{y}$.  We refer to $R'$ as an anchor of $R$.
\end{definition}

\begin{algorithm}[t]
\caption{$\textsc{Reduction}(Q=(\V,\E,\bar{y}))$}
\label{alg:trim}
$\E_e \gets \{\text{all ears in }\E\}$\;
\While{there exists $R(\bar{x}) , R'(\bar{x}') \in \E$ s.t. $R \in \E_e$ and $R'$ is its anchor}{
    $\mathcal{E} \gets \mathcal{E} - \{R\}$
}
\Return $\E$\;
\end{algorithm}

The \emph{skeleton} of a query $Q$, denoted $\E_*$, is the subset of relations remaining after recursively removing an ear from $\E$ until no further removal is possible. This reduction is summarized in Algorithm~\ref{alg:trim} and is similar to the GYO algorithm, except that the set of ears is computed once at the start and held fixed while removing them from $\E$. As a consequence, the skeleton is guaranteed to be non-empty, as removing an ear requires an anchor in the remaining query. We further define the concept of residual query for a given free-connex query:
\begin{definition}[Residual Query]
    Given a free-connex query $Q$ with its skeleton $\E_*$, where $\V_* = \{x\mid \exists R(\bar{x}) \in \E_*, x \in \bar{x}\}$ is the set of attributes in $\E_*$, and $\bar{y}_* = \{x\mid x \in \bar{y} \land \exists R(\bar{x}) \in \E_*, x \in \bar{x}\}$ is the set of output attributes in $\E_*$,  the residual query $Q_*$ for $Q$ is $\E_*$-induced query $Q_* \gets (\V_*, \E_*, \bar{y}_*)$. 
\end{definition}
It is possible that a pair of relations $R_i, R_j$ is both an ear relation and the anchor of the other, making the skeletons, and thus the residual queries, not unique.  However, the underlying hypergraphs for different residual queries should be isomorphism and the hardness of these residual queries remains the same:
\begin{lemma}
\label{lem:2.9}
    Let $Q_1 = (\V_1, \E_1, \bar{y}_1)$ and $Q_2 = (\V_2, \E_2, \bar{y}_2)$ are two residual queries for a given free-connex query $Q$, then for every $R(\bar{x}) \in \E_1$, it is either (1) $R(\bar{x}) \in \E_2$, or (2) there exists $R'(\bar{x}') \in \E_2$, such that $R'(\bar{x}') \notin \E_1$, and $\bar{x} \setminus \V_\bullet = \bar{x}' \setminus \V_\bullet$.
\end{lemma}

\begin{proof}

By definition, every relation that is not an ear must remain in the skeleton. Hence, any ear with a non-ear anchor must be removed. Now consider two residual queries $Q_1$ and $Q_2$. If $R(\bar{x}) \in \E_1$ but $R(\bar{x}) \notin \E_2$, then $R$ must be an ear of the original query $Q$ whose anchors are also ears of $Q$. Since $R$ is removed in $Q_2$ but remains in $Q_1$, there must exist an anchor $R' \in \E_2$ of $R$ such that $R' \notin \E_1$. By the same argument, because $R' \notin \E_1$, it must have an anchor $R'' \in \E_1$. The anchor relation is transitive: if $R_i$ is an anchor of $R_j$ and $R_k$ is an anchor of $R_i$, then $R_k$ is also an anchor of $R_j$. Therefore, we must have $R'' = R$; otherwise, by transitivity, $R$ would also have an anchor in $\E_1$ different from itself, and thus should have been removed from $\E_1$, a contradiction. It follows that $R$ and $R'$ are anchors of each other, and hence $\bar{x} \setminus \V_\bullet^1 = \bar{x}' \setminus \V_\bullet^2$.
\end{proof}

\begin{example}
\label{ex:2}
    Considering $Q_1$ from Example~\ref{ex:1}.  Relation $R_1$, $R_6$, $R_8$, $R_9, R_{10}$ are ears, with $R_2$ be the anchor of $R_1$, $R_5$ be the anchor of $R_6$, $R_4$ be the anchor of $R_8 $ and $ R_9$, $R_7$ be the anchor of $R_{10}$.  Its residual query is 
    \[
        Q_1'(x_3, x_4, x_9) \gets R_2(x_2, x_3), R_3(x_3, x_4, x_5, x_6), R_4(x_4, x_7), R_5(x_7, x_8), R_7(x_4, x_5, x_9).
    \]
\end{example}

\begin{definition}
A free-connex query $Q$ is weak-q-hierarchical if $Q$ is not q-hierarchical and its residual query $Q_*$ is q-hierarchical. 
\end{definition}

\rwb{
\paragraph{Free-connex join tree height} The classification of free-connex queries is closely tied to the height of their free-connex join trees~\cite{wang2023change,pods25}.

\begin{definition}
    The \emph{height} of a generalized join tree $\T$ is the maximum count of input relations on any path from a leaf to the root without considering the generalized relations. 
\end{definition}

\begin{lemma}[\cite{pods25}, Lemma~2.2 and Lemma~3.7]
Let $Q$ be a free-connex conjunctive query. $Q$ is q-hierarchical if and only if it admits a height $1$ free-connex join tree, and $Q$ is weak-q-hierarchical if and only if it admits a height $2$ free-connex join tree.
\end{lemma}

In this work, we refine the classification of conjunctive queries by extending the notions of height. When the query is free-connex, we show that a query has height $k$ exactly when it admits a height $k$ free-connex join tree but admits no height $k-1$ free-connex join tree. Under this view, the q-hierarchical and weak-q-hierarchical classes correspond precisely to the height $1$ and height $2$ queries, respectively.

}

\subsection{$k$-Cycle Detection Problem}

Finding cycles in graphs is a fundamental problem in algorithmic graph theory. Given a directed graph $G = (V, E)$ with $n$ vertices and $m$ edges, a $k$-cycle is a sequence of $k$ distinct vertices $(v_1, v_2, \ldots, v_k)$ such that $(v_i, v_{i+1}) \in E$ for every $1 \le i < k$, and $(v_k, v_1) \in E$. The $k$-cycle detection problem asks whether a given graph contains a $k$-cycle. This problem is NP-complete in general. For example, when $k = n$, it reduces to the Hamiltonian cycle problem. On the other hand, it is fixed-parameter tractable (FPT) when parameterized by $k$, meaning it is polynomial-time solvable for each fixed $k = O(1)$.

The $k$-cycle detection problem also serves as a core problem in fine-grained complexity. In particular, the special case $k=3$ (triangle detection) has been extensively used to establish conditional lower bounds for database query processing~\cite{conjunctiveLB,comparison,difference}. However, beyond triangles, the general $k$-cycle problem is rarely discussed in the context of database queries.  
\citet{lincoln_et_al:LIPIcs.ITCS.2020.11} have proved the lower bounds for detecting odd cycles and even cycles on sparse graphs based on the combinatorial $k$-clique conjecture.

\begin{conjecture}[Combinatorial $k$-Clique Conjecture~\cite{DBLP:conf/focs/AbboudBW15a}] There is no \textbf{combinatorial} algorithm that can detect a $k$-clique in a graph with $n$ vertices and $O(n^2)$ edges in $O(n^{k-\epsilon})$ time, for any constant $\epsilon > 0$ and any integer $k \ge 3$
\end{conjecture}


\begin{theorem}[Odd Cycle~\cite{lincoln_et_al:LIPIcs.ITCS.2020.11}, Theorem 15]
\label{thm:oddcycle}
    Given a directed graph $G = (V, E)$ with $|V| = n$ and $|E| = m$, and $m = n^{1+2/(k-1)}$
    , no \textbf{combinatorial} algorithm can detect the existence of an odd $k$-cycle for any odd integer $k \ge 3$ in $G$ in time $O(m^{2k/(k+1)-\epsilon})$ for any constant $\epsilon > 0$, assuming combinatorial $k$-clique conjecture.
\end{theorem}

\begin{theorem}[Even Cycle~\cite{lincoln_et_al:LIPIcs.ITCS.2020.11}, Corollary 16]
\label{thm:evencycle}
    Given a directed graph $G = (V, E)$ with $|V| = n$ and $|E| = m$, and $m = n^{k/(k-2)}$, no \textbf{combinatorial} algorithm can detect the existence of a even $k$-cycle for any even integer $k \ge 4$ in $G$ in time $O(m^{(2k-2)/k-\epsilon})$ for any constant $\epsilon > 0$, assuming combinatorial $k$-clique conjecture.
\end{theorem}

Note that both the odd and even cycle lower bounds assume sparse graphs. This restriction is necessary to prove lower bounds for database queries, as we measure complexity in terms of the input size $O(|D|)$ rather than the domain size. 
On the upper bound side, \citet{DBLP:journals/algorithmica/AlonYZ97} gave an algorithm running for detecting a $k$-cycle with running time $O(m^{2-(1+\lceil k/2\rceil^{-1})/(k+1)})$.  Later,  \citet{FindingCycleUpperBound} improved the running time for finding even cycles in undirected graphs to $O(m^{2k/(k+2)})$. Furthermore, \citet{lincoln2018tight} provided a reduction that connects the complexity of detecting an undirected even cycle to that of detecting a directed odd cycle.  

We will now briefly review relevant techniques from these works that are useful for our proofs.
\paragraph{Color-Coding} The primary challenge in detecting $k$-cycles is ensuring that all $k$ vertices in the cycle are distinct. Color-coding is a randomized technique introduced by \citet{colorcoding} that addresses this challenge. The core idea is to assign each vertex $v \in V$ a random color from a set of $k$ colors. A \emph{colorful} $k$-cycle, one in which all $k$ vertices have distinct colors, is by definition a simple $k$-cycle.  If the graph contains a $k$-cycle, then with some probability, the vertices of that cycle will be colorful under the random color assignment.  By repeating this randomized process a sufficient number of times, any existing $k$-cycle can be found with high probability.

This randomized procedure can be derandomized using a $(n, k)$-perfect hash family~\cite{DBLP:conf/focs/NaorSS95}.
\begin{definition}[$(n, k)$-Perfect Hash Family]
A $(n, k)$-perfect hash family is a set of hash functions $H$ mapping a universe of $n$ items to a range of $k$ values such that for every subset $S$ of items with $|S| = k$, there exists at least one hash function $f \in H$ that is injective on $S$ (i.e., $f$ maps the $k$ distinct items in $S$ to $k$ distinct values).
\end{definition}
Given such a family $H$, one can iterate over each $f \in H$ to color the vertices of $G$ and then run the cycle-detection algorithm on the colored graph. The family $H$ guarantees that for every $k$-cycle in $G$, at least one hash function $f \in H$ makes the cycle colorful. A construction for $H$ exists such that $|H| = k^{O(1)}\log(n)$~\cite{DBLP:conf/focs/NaorSS95}, ensuring the derandomized algorithm remains FPT. In this paper, we treat color-coding and perfect hashing as black boxes that we can employ when needed.

\subsection{(Generalized) Online Matrix-Vector Multiplications}

Online matrix-vector multiplication (OMv) \cite{abboud2014popular} is a fundamental problem for understanding the hardness of online problems.  Formally, given an $n \times n$ Boolean matrix $M$ and a sequence of $n$-dimensional vectors $\{v_1, \cdots, v_n\}$ is given in order.  The task is to answer the matrix-vector multiplication $Mv_i$ for every $v_i$ before receiving the next vertex $v_{i+1}$.  This leads to the well-known \emph{Online Matrix-Vector Conjecture}.

\begin{conjecture}[OMv Conjecture]
    There is no algorithm for the OMv problem with a sequence of $n$ vectors that runs in time $O(n^{3-\epsilon})$ for any constant $\epsilon > 0$. 
\end{conjecture}

The conjectured hardness result implies lower bounds on a variety of dynamic problems, which also lead to fine-grained complexity results for answering conjunctive queries \cite{pods25,berkholz17:_answer,kara2020maintaining}.

\citet{generalizedOMv} generalized the problem to a tensor version, and proposed the parameterized OuMv$_k$.  Given a tensor $M \in [n]^k$, for every round $i$, a sequence of $k$ n-dimensional vectors $\{u^{(1)}, u^{(2)}, \cdots, u^{(k)}\}$ is given, the task is to answer whether $u^{(1)} \times u^{(2)} \times \cdots \times u^{(k)}$ has a non-empty intersection with $M$, before receiving the next round of vectors. 

\rwa{
\begin{conjecture}[OuMv$_k$ Conjecture]
    There is no algorithm for the OuMv$_k$ problem with $n^{\gamma}$ queries with linear pre-processing time and $O(n^{\gamma+k-\epsilon})$ total query time for any $\gamma, \epsilon > 0$.
\end{conjecture}
}

In this paper, we apply the generalized OuMv$_k$ conjecture to provide tighter lower bounds for maintaining CQ under updates.

\section{Maintenance Hardness when Scaling Out}
\label{sec:height}
From the examples of q-hierarchical and weak-q-hierarchical queries, we observe a pattern: when scaling out the query by adding additional relations, the hardness of maintaining the query increases.  In this section, we investigate how maintenance hardness changes as we scale out.

\subsection{Cycle Detection and Conjunctive Queries}
\label{sec:hardqueries}
We first start by drawing the connection between answering conjunctive queries with $k$ cycle detection, which identifies three classes of hardness queries, $Q_{2k-1}$, $Q_{2k}$, and $Q'_{k}$. 

\begin{lemma}
\label{lem:oddfulljoin}
    Consider the query with arbitrary output attributes $\bar{y}$:
    \[ 
        Q_{2k-1}(\bar{y}) \gets R_1(x_1), R_2(x_1, x_2), \cdots, R_{2k-2}(x_{2k-3},x_{2k-2}), R_{2k-1}(x_{2k-2}).
    \]
    If $Q_{2k-1}$ can be maintained in amortized $O(|D|^{(k-1)/k-\epsilon})$ time while supporting $O(|D|^{1-\epsilon})$-delay enumeration for any constant $\epsilon > 0$, \rwa{the combinatorial $k$-clique conjecture failed. }
\end{lemma}

\begin{proof}
    Given a graph $G = (V, E)$ with $m$ edges and $n$ vertices with $n = m^{(k-1)/k}$, we first find an $(n, 2k-1)$-perfect hash family $H$, and for every $f \in H$, we color every vertex using the hash function $f$ and thus partition the graph into $2k-1$ disjoint subgraphs as $V = V^{1} \cupdot V^{2} \cupdot \cdots \cupdot V^{2k-1}$, where $V^{x} \gets \{v \in V|f(v) = x\}$.  Let $P$ be the set of all $(2k-1, 2k-1)$-permutations, and let $p$ be any of such permutation, for any $f, p$, we create the update sequence as follows: 
    \begin{enumerate}[leftmargin=*]
    \item We assign $R_i = \{(u, v) \in E|f(u)=p(i-1), f(v) = p(i)\}$  and we insert all such edges into $R_i$ for all $i \in [2, 2k-2]$\;
    \item For every $v \in V$ such that $f(v) = p(2k-1)$, we insert all $\{(u)|(v, u) \in E, f(u) = p(1)\}$ into $R_1$, and all $\{(u)|(u, v) \in E, f(u) = p(2k-2)\}$ into $R_{2k-1}$.  
    \item We try to enumerate $Q_{2k-1}$ on the current database instance.  If there exists any results, then there exists a $2k-2$ path $x_1, x_2, \cdots, x_{2k-2}$, because the color coding guarantee that for any pair of $x_i, x_j$, $x_i \neq x_j$.  In addition, we know that there exist two edges $(v, x_1)$ and $(x_{2k-2}, v)$ with $v$ having a different color from any $x_i$.  Therefore, $(v, x_1, \cdots, x_{2k-2})$ form a length $2k-1$-cycle, and we can terminate by reporting $2k-1$ cycle detected.
    \item Otherwise, we delete all tuples from $R_1$ and $R_{2k-1}$, and repeat (2) - (4) until no further $v$ exists.  If that is the case, we can conclude a $2k-1$ cycle does not exist for the given $f, p$.
    \end{enumerate}
    The above update sequence inserts every edge into the database once, and deletes every edge from the database at most once, making a total of $O(m)$ updates; therefore, if such a query can be maintained in $O(|D|^{(k-1)/k-\epsilon})$ time, maintaining the entire update sequence takes a total running time of $O(m^{(k-1)/k-\epsilon}\times m) = O(m^{(2k-1)/k-\epsilon})$ time.  By setting $k' = 2k-1$, $O(m^{(2k-1)/k-\epsilon})=O(m^{2k'/(k'+1)-\epsilon})$.  Meanwhile, a total of $O(n)$ enumeration operations will be issued during the above procedures. If the enumeration delay is $O(m^{1-\epsilon})$, then for the enumeration procedure, we would need to spend $O(m^{1-\epsilon}\times m^{(k-1)/k}) = O(m^{2k'/(k'+1)-\epsilon})$ time, summing up the two parts resulting in a total running time of $O(m^{2k'/(k'+1)-\epsilon})$ for the above procedures.

    Note that we also need to verify all possible color coding and permutations to ensure we don't miss any cycles in the graph.  The total running time for maintaining all $f \in H$ and $p \in P$ takes 
    \[O(\underbrace{\log m}_{\textrm{size of }H}\times \underbrace{(k')!}_{\textrm{Number of permutations}}\times \underbrace{m^{2k'/(k'+1)-\epsilon}}_{\textrm{Maintenance cost for each }p, f}) = \underbrace{O(m^{2k'/(k'+1)-\epsilon'})}_{k' = O(1),\log m \text{ suppressed in } \epsilon'}.\]

    Once finished, we can determine if the given graph $G$ contains a $k' = 2k-1$ cycle.  \rwa{Since the above procedure is combinatorial, it provides a combinatorial algorithm for solving the odd cycle detection problem.  However, Theorem~\ref{thm:oddcycle} suggests that there is no combinatorial algorithm that can detect such a cycle in time $O(m^{2k'/(k'+1)-\epsilon})$ unless the combinatorial $k$-clique conjecture failed.  Thus, the existence of such a maintenance method also suggests the conjecture failed. }
\end{proof}

\begin{lemma}
    Consider the query with arbitrary output attributes $\bar{y}$:
    \[
        Q_{2k}(\bar{y}) \gets R_1(x_1), R_2(x_1, x_2), \cdots, R_{2k-1}(x_{2k-2}, x_{2k-1}), R_{2k}(x_{2k-1}).
    \]
    If $Q_{2k}$ can be maintained in amortized $O(|D|^{(k-1)/k-\epsilon})$ time while supporting $O(|D|^{1-\epsilon})$-delay enumeration for any constant $\epsilon > 0$, \rwa{the combinatorial $k$-clique conjecture failed. }
\end{lemma}

\begin{proof}
    Given a graph $G = (V, E)$ with $m$ edges, $n = m^{(k-1)/k}$ and $k' = 2k$, we construct the same update sequence as for the odd cycle case.  Now the total update time would be $O(m^{(2k-1)/k-\epsilon})= O(m^{(2k'-2)/k'-\epsilon})$ and the enumeration time would also be $O(m^{(2k-1)/k-\epsilon})= O(m^{(2k'-2)/k'-\epsilon})$ for any color coding $f \in H$ and permutation $p$, which would return a \rwa{combinatorial algorithm for detecting even cycle in time $O(m^{(2k'-2)/k'-\epsilon})$.  Combining with Theorem~\ref{thm:evencycle}, the existence of the maintenance algorithm suggests the combinatorial $k$-clique conjecture failed.}
\end{proof}

\begin{lemma}
    Consider the query 
    \[
        Q_{k}(x_k) \gets R_1(x_1), R_2(x_1, x_2), \cdots, R_k(x_{k-1}, x_{k}).
    \]
    If $Q_{k}$ can be maintained in amortized $O(|D|^{(k-1)/k-\epsilon})$ time while supporting $O(|D|^{1/k-\epsilon})$ delay enumeration for any constant $\epsilon > 0$, \rwa{the combinatorial $k$-clique conjecture failed.}
\end{lemma}

\begin{proof}
    Given a graph $G = (V, E)$ with $m$ edges, we maintain two $Q_{k}(x_k)$ queries as 
    \[
        \begin{aligned}
        Q_{k}(x_k) &\gets R_1(x_1), R_2(x_1, x_2), \cdots, R_k(x_{k-1}, x_k)\\
        Q_{k}'(x_k)&\gets R_{2k-1}(x_{2k-2}), R_{2k-2}(x_{2k-3}, x_{2k-2}), \cdots, R_k(x_{k-1}, x_k).
        \end{aligned}
    \]
    The update sequence is the same as the proof for Lemma~\ref{lem:oddfulljoin}; the only difference is in the enumeration procedure.  Once we finish every round of updates, we enumerate both $Q_{k}$ and $Q_{k}'$ and check if $Q_{k} \cap Q_{k}'$ is empty.  If it is not empty, then a $k'=2k-1$ cycle is detected.  The output size is at most $n = m^{(k-1)/k}$.  Therefore, the delay should be $O(|D|^{1/k-\epsilon})$ for some $\epsilon > 0$ to guarantee strictly smaller overhead than re-evaluation.  When the amortized update time is $O(|D|^{(k-1)/k-\epsilon})$, the total running time becomes $O(m^{(2k-1)/k-\epsilon}) = O(m^{2k'/(k'+1)-\epsilon})$, \rwa{and the above procedure provides a combinatorial algorithm for odd $k'$ cycle detection problem .  Combining with Theorem~\ref{thm:oddcycle}, the existence of the maintenance algorithm suggests the combinatorial $k$-clique conjecture failed.}
\end{proof}

\subsection{Chordless Path and Heights} 
The queries presented in Section~\ref{sec:hardqueries} suggest that the hardness increases when more relations are involved.  To quantify the influence, we introduce a structure parameter on the relational hypergraph, \emph{height}, for any CQ $Q$, based on the chordless path on the relation hypergraph.

\rwb{Given an undirected graph $G = (V, E)$, a \emph{chordless path}~\cite{HAAS2006360} on $G$ is a sequence of distinct vertices $P = \{v_1, \cdots, v_{n+1}\}$ and an edge sequence $P_e = \{e_1, \cdots, e_{n}\}$, such that $e_i = (v_i, v_{i+1})$ for every $i \in [1, n]$, and the subgraph $G_P \gets (V_p = P, E_p = \{(v_i, v_j) \in E, v_i, v_j \in P\})$ induced by $P$ is the path itself. In other words, for any $i, j \in [1, n+1]$ with $|i-j|>1$, $(v_i, v_j) \notin E$.   

While the path length in graphs is naturally defined by the number of edges ($n$), generalizing this to hypergraphs requires caution because a single hyperedge may contain an arbitrary number of vertices, and given the sequence $P$, the edge sequence $P_e$ may not be unique. We propose the following generalization:}

\begin{definition}[Chordless Path on Hypergraph]
Given a hypergraph $\H = (\V, \E)$, we define the chordless path on the hypergraph $\H$ to be a sequence of distinct vertices $P$, and an \emph{edge sequence} $P_e$ of the chordless path $P$ to be a sequence of hyperedges.  A chordless path $P$ is
\begin{itemize}[leftmargin=*]
    \item a length $1$ chordless path if $P \gets \{\emptyset\}$, and $P_e \gets \{e_1\}$ as an edge sequence of $P$ with an arbitrary hyperedge $e_1$\;
    \item a length $2$ chordless path if $P \gets \{v_1\}$ and there exists two hyperedges $e_1, e_2$ with $v_1 \in e_1 \cap e_2$.  All such pair of hyperedges form an edge sequence $P_e \gets \{e_1, e_2\}$ of $P$\;
    \item a length $k$ chordless path for $k \ge 3$ if $P \gets \{v_1, \cdots, v_{k-1}\}$ such that
    \begin{enumerate}[leftmargin=*]
    \item \textbf{(Adjacency)} For every $i \in [1, k-2]$, there exists a hyperedge $e_{i+1} \in E$ with both $v_i, v_{i+1} \in e$;
    \item \textbf{(No shortcuts)} For any $i, j \in [1, k-1]$ with $|i-j| > 1$, there is no $e \in E$ with both $v_i,v_j \in e$;
    \item \textbf{(Endpoints)} there exist $e_1, e_k \in E$, such that \rwb{$v_1 \in e_1$ and $v_2 \notin e_1$, and $v_{k-1} \in e_k$ and $v_{k-2} \notin e_k$}.
    \end{enumerate}
    Any $P_e \gets \{e_1, \cdots, e_k\}$ that satisfies the above conditions is an edge sequence of $P$.  Note that all hyperedges in an edge sequence are also distinct due to the ``No shortcuts" requirement, and such an edge sequence may not be unique for a given hypergraph.
\end{itemize}
A chordless path $P$ is \emph{maximum} if there is no strictly longer chordless path $P'$ with $P \subsetneq P'$.  \rwb{This definition aligns with the standard graph definition when all hyperedges are standard edges, and the sequence $P$ is defined as the path's internal vertices $v_2, ..., v_n$. The `Adjacency' and `No shortcuts' requirements are both naturally generalized from the standard graph. However, only one vertex is kept for the two `endpoint' hyperedges due to the possible existence of unary hyperedges.  }

In addition, for non-full and non-Boolean queries, we further define the concept of a q-chordless path \rwb{as a chordless path with only non-output attributes in $P$, while one `endpoint' hyperedge contains output attributes. This generalizes the notions of q-cores and hook cores from \cite{berkholz17:_answer,pods25}.} The length $1$ chordless path is also a q-chordless path, with the edge sequence $\{e_1\}$ satisfying $e_1 \cap \bar{y} \neq \emptyset$ and $e_1 \setminus \bar{y} \neq \emptyset$, i.e., $e_1$ that contains both output and non-output attributes.  A length $k$ ($k \ge 2$) q-chordless path is a chordless path $P \gets \{v_1, \cdots, v_{k-1}\}$ on $\H$, such that
\begin{enumerate}[leftmargin=*]\setcounter{enumi}{4}
    \item \textbf{(Non-output Vertices)} for all $x_i \in P$, $x_i \notin \bar{y}$, i.e. all attributes in $P$ are not output attributes;
    \item \textbf{(Output Endpoint)} there exists a hyperedge $R(\bar{x})$ with $x_1 = \bar{x} \cap P$, and $\bar{x} \cap \bar{y} \neq \emptyset$, i.e., $\bar{x}$ contains both output attributes and an endpoint $x_1$ of $P$.  The edge sequence of a q-chordless path must contain one of such hyperedges $R$.
\end{enumerate}
\end{definition}

\begin{example}
\label{ex:3}
    Consider $Q_1$ from Example~\ref{ex:1}.  We can find a chordless path $P \gets \{x_2, x_3, x_4, x_7, x_8\}$ of length $6$ on its relational hypergraph (which is marked in red in Figure~\ref{fig:hypergraph}) with $P_e \gets \{R_1, R_2, R_3, R_4, R_5, R_6\}$.  Note that $x_1$ cannot serve as one endpoint because there is no relation that contains only $x_1$ but no other $x \in P$.  Additionally, we cannot include $x_6$ and $x_5$ in the chordless path between $x_3$ and $x_4$, as $R_3$ can serve as a shortcut.

    Furthermore, we can find a length $3$ q-chordless path $P \gets \{x_7, x_8\}$ with edge sequence $P_3 \gets \{R_4, R_5, R_6\}$, such that $R_4$ is the output endpoint with output attribute $x_4$.
\end{example}

With the chordless path, we can now define the height of a query:
\begin{definition}
    Given a conjunctive query $Q$, let $\H = (\V, \E)$ be $Q$'s relational hypergraph, query $Q$ is a height $k$ query if $\H$ contains no chordless path with length greater than $2k$, $\H$ contains no q-chordless path with length greater than $k$, and 
    \begin{itemize}[leftmargin=*]
        \item $\H$ contains a length $2k$ or length $2k-1$ chordless path; and/or 
        \item $\H$ contains a length $k$ q-chordless path.
    \end{itemize}
\end{definition}

\rwb{
\begin{lemma}
\label{lem:height-hierarchical}
A conjunctive query $Q$ is height $1$ if and only if it is q-hierarchical.
\end{lemma}
}

When limiting the query to free-connex, we can relate the query height to the free-connex join tree height.  The key idea is that removing all ear nodes from $Q$ produces a residual query $Q'$ whose height drops by one.  Moreover, a free-connex join tree for $Q'$ can be lifted to a free-connex join tree for $Q$ by re-attaching the removed ears and increasing the tree height by one. Since a q-hierarchical (height $1$) query admits a height $1$ free-connex join tree, an induction on the height shows that every height $k$ query admits a height $k$ free-connex join tree.  

On the other hand, every chordless path in the relational hypergraph of $Q$ must be realized within any free-connex join tree, which rules out the existence of a height $k-1$ free-connex join tree for a height $k$ query. Putting the two directions together yields the following lemma with the detail proof given in Appendix~\ref{app:heights}.

\begin{lemma}
\label{thm:tree-height}
    Given a height $k$ free-connex query $Q$, $Q$ admits a free-connex join tree of height $k$, but not of height $k-1$.
\end{lemma}

\subsection{Proof of Theorem~\ref{thm:height}}

After identifying a length $2k$ or $2k-1$ chordless path in height $k$ query, or a length $k$ q-chordless path in the non-full and non-boolean height $k$ query, we can now draw a connection between height $k$ queries and cycle detection problem using the chordless path and its induced query, and prove Theorem~\ref{thm:height} for the hardness of maintaining a height $k$ query. 

Given a height $k$ query $Q = (\V, \E, \bar{y})$, we can then simulate $Q_{2k-1}$, $Q_{2k}$, or $Q'_{k}$.  Without loss of generality, we use $Q_{2k-1}$ as an example.  Given an update sequence $S$ for $Q_{2k-1}$, we first identify a length $2k-1$ chordless path $P = \{x_1, \cdots, x_{k-2}\}$ from $Q$, which corresponds to all attributes in $Q_{2k-1}$ . For every $v \in \V$ such that $v \notin P$, we let their domain be a special value $\perp$.  We can then create an equivalent update sequence $S'$ for $Q$, such that for every $R(\bar{x}) \in \E$, we create the updates for $R$ by the following rules:
\begin{itemize}[leftmargin=*]
    \item \textbf{$\bar{x} \cap P = \emptyset$.} We insert the tuple $(\perp, \cdots, \perp)$ for $R$ at the beginning of all updates;
    \item \textbf{$\bar{x} \cap P = \{x_i\}$, where $x_i$ is not an endpoint of $P$.} We insert the tuple $(u, \perp, \cdots, \perp)$ by assigning $x_i = u$ for all $u \in \dom(x_i)$ and all other attributes to $\perp$ at the beginning of all updates;
    \item \textbf{$\bar{x} \cap P = \{x_i\}$, where $x_i$ is an endpoint of $P$.}  When there is an update $(u)$ to $R_1$ (or $R_{k-1}$) for $Q_{2k-1}$, we update $(u, \perp, \cdots, \perp)$ from $R$ with $x_1 = u$ (or $x_{k-2} = u$) and all other attributes to $\perp$;
    \item \textbf{$\bar{x} \cap P = \{x_i, x_{i+1} \}$}.  When there is an update $(u, v)$ to $R_{i}$, we insert or delete $(u, v, \perp, \cdots, \perp)$ from $R$ with $x_i = u$ and $x_{i+1} = v$ and all other attributes to $\perp$.
\end{itemize}
Due to the property of a chordless path, there is no $R$ with more than two attributes in $P$, or $\bar{x} \cap P = \{x_i, x_j\}$ with $|j-i| > 1$.  It is not hard to see that there is a one-to-one mapping for the query result $Q$ and the $Q_{2k-1}$, and the result can be easily generalized to $Q_{2k}$ and $Q'_{k}$.  Therefore, given any update sequence $S$ for $Q_{2k-1}$, $Q_{2k}$ or $Q'_{k}$, we can create an equivalent update sequence $S'$ on $Q$ as long as $Q$ has a length $2k-1$, $2k$ chordless path, or a length $k$ q-chordless path, thus maintaining a height $k$ query should be as hard as maintaining $Q_{2k-1}$, $Q_{2k}$ or $Q'_{k}$.

\section{Maintenance Hardness when Scaling Up}
\label{sec:dimension}

In Section~\ref{sec:height}, we analyzed the parameterized hardness of extending a query with additional relations. It is natural to ask whether two queries of the same height can still differ in maintenance difficulty due to the structure of their join attributes or unique output attributes\footnote{For unique non-output attributes, we can project those attributes out without affecting the query result.}. To capture this effect, we introduce a second parameter, the \emph{dimension} of a query. To define the dimension, we first formalize the notion of a \emph{connected subset} on a query.

\begin{definition}[Connected Subset]
    Given a conjunctive query $Q = (\V, \E, \bar{y})$, let $\V_c \subseteq \V \setminus \V_\bullet$ be a subset of join attributes. We define the induced active relations $\E_c = \{R(\bar{x}) \in \E \mid \bar{x} \cap \V_c \neq \emptyset\}$.  To refine the definition, when $\V_c = \{\emptyset\}$, every $R_i(\bar{x}_i) \in \E$ can be the active relation $\E_c = \{R_i\}$.    
    
    We say $\V_c$ is a \emph{connected subset} if the primal graph $G(\V_c)$ is connected.  
    Let $\V(\E_c) = \cup_{R(\bar{x}) \in \E_c} \bar{x}$ denote the set of attributes appearing in $\E_c$. We define the set of \emph{active attributes} as $\V_\star = \V(\E_c) \setminus \V_c$.
\end{definition}

\begin{definition}[Dimension]
\label{def:dimension}
    Let $Q = (\V, \E, \bar{y})$ be a conjunctive query. Let $\V_c \subseteq \V$ be a connected subset with induced relations $\E_c$. Let the remaining relations be $\E_r = \E \setminus \E_c$.
    
    A subset of relations $\mathcal{N} \subseteq \E_r$ is called a \emph{distinct neighbor set} if there exists an injective mapping $\pi: \mathcal{N} \to \V_\star$ such that for every $R(\bar{x}) \in \mathcal{N}$, $\pi(R) \in \bar{x} \cap \V_\star$. We refer to the image of this mapping, $\V_k(\mathcal{N}) = \{\pi(R) \mid R \in \mathcal{N}\}$, as a \emph{key set}. Note that multiple key sets may exist for a fixed $\mathcal{N}$.
    
    The \emph{local dimension} of $\V_c$, denoted $\dim(\V_c)$, is defined as:
    \[
        \dim(\V_c) = \max_{\substack{\mathcal{N} \subseteq \E_r, \ \V_k(\mathcal{N})}} 
        \begin{cases}
            |\V_k(\mathcal{N})| & \text{if } \V_c \cap \bar{y} \neq \emptyset \text{ or } \V_k(\mathcal{N}) \subseteq \bar{y}, \\
            |\V_k(\mathcal{N}) \cup (\bar{y} \cap \V_{\star})| & \text{otherwise.}
        \end{cases}
    \]
    Intuitively, if $\V_c$ contains output attributes, or $\V_k(\mathcal{N})$ contains only output attributes, we only count the contribution of $\V_k(\mathcal{N})$. Otherwise, we additionally count any other active output attributes.
    
    The \emph{dimension} of the query $Q$ is the maximum local dimension over all connected subsets:
    \[
        \dim(Q) = \max_{\V_c \subseteq \V} \dim(\V_c).
    \]
\end{definition}

We now establish a connection between the query dimension and the hierarchical structure.

\begin{lemma}
    \label{lem:hierarchical}
    A CQ is q-hierarchical if and only if $\dim(\V_c) = 1$ for all $\V_c \subseteq \V$.
\end{lemma}

\rwb{Combining with Lemma~\ref{lem:height-hierarchical}, we are able to show the equivalence between dimension-$1$ queries/height $1$ queries/q-hierarchical queries and prove Proposition~\ref{thm:equiv}.}

\begin{example}
    Consider query $Q_1$ from Example~\ref{ex:1}. It is a dimension-$4$ query with $\dim(\{x_5\}) = 4$. Here, the connected set is $\V_c = \{x_5\}$ and the induced relations are $\E_c = \{R_3, R_7\}$. Figure~\ref{fig:bipartite} illustrates a distinct neighbor set $\mathcal{N} = \{R_2, R_4, R_8, R_{10}\}$ maximizing the dimension. The corresponding key set $V_k(\mathcal{N})$ includes $\{x_3, x_4, x_6, x_9\}$.
\end{example}

\subsection{Star Queries}
We link a dimension-$d$ conjunctive query to the hardness of the Generalized Online Matrix-Vector Multiplication (OuMv$_k$) problem using the construct of \emph{Star Queries}, where there exists a \emph{center} relation that is connected to a set of \emph{satellite} relations.  All these satellite relations do not share any common attributes, whereas the center relation does share common attributes with each of them.  Without loss of generality, we consider star queries in the following form:

\begin{definition}[Star Queries]
\label{def:stars}
    A query $Q = (\V, \E, \bar{y})$ is a \emph{star query} if it is of the form:
    \[
        Q(x_{d-j+1}, \dots, x_d) \leftarrow R_1(x_1, \cdots, x_d), R_2(x_1), \dots, R_{k+1}(x_k)
    \]
    for any $0 \le d-j \le k \le d$, where $R_1$ is the \emph{center} relation, $R_2, \dots, R_{k+1}$ are \emph{satellite} relations, $d$ is the dimension of the query and $j$ is the number of output attributes.  There are two special cases: $j = d, k=d$ (Full query) or $j = 0, k = d$ (Boolean query with no output attributes).   
\end{definition}

Star-shaped schemas are ubiquitous in data warehouse design: the center relation $R_1$ is the fact table, while all other satellite relations are dimension tables.  
In practice, these joins are usually between primary and foreign keys.  Although we do not impose key constraints in the query definition, our lower bound instances respect them; meanwhile, prior work shows that, under restricted update sequences such as first-in-first-out, imposing key constraints can change the hardness of maintaining these sequences~\cite{wang2020maintaining}. 

Any star query is a height $2$ query: a height $2$ free-connex join tree $\T$ can be obtained by taking $R_1$ as the root node and all satellites as leaf nodes attached to $R_1$.  However, we can show that their lower bound can exceed the $O(\sqrt{|D|})$ bound established in Theorem~\ref{thm:height}:

\begin{lemma}
    Consider the star query:
    \(
        Q_d() \gets R_1(x_1, \dots, x_d), R_2(x_1), \dots, R_{d+1}(x_d).
    \)
    This query has a dimension of $d$. Unless the OuMv$_k$ Conjecture fails, $Q_d$ cannot be maintained with amortized $O(|D|^{(d-1)/d - \epsilon})$ time and $O(1)$-delay enumeration for any constant $\epsilon > 0$.
\end{lemma}

\begin{proof}
    We reduce the OuMv$_k$ problem to maintaining $Q_d$. Let $M$ be an $d$-dimensional Boolean tensor of size $n \times \dots \times n$, and let $u_1, \dots, u_d$ be a stream of $n^\gamma$ vector sets. Since $n = |D|^{1/d}$, when $\gamma \ge d-1$, the total input size is bounded by $(|D|^{1/d})\times(|D|^{1/d})^\gamma = |D|^{(\gamma+1)/d}$.  Given an OuMv$_k$ instance, we can solve the problem with the following reduction:
    \begin{enumerate}[leftmargin=*]
        \item Encode tensor $M$ into $R_1$: insert tuple $(t_1, \dots, t_d)$ into $R_1$ if and only if $M[t_1, \dots, t_d] = 1$.
        \item For each round $j$ of vector updates and queries:
        \begin{itemize}[leftmargin=*]
            \item For each $i \in [1, d]$, update $R_{i+1}$ to represent vector $u_i^{(j)}$ by inserting $(x_i)$ if $u_i^{(j)}[x_i] = 1$.
            \item Enumerate the result of $Q_d$. The result is non-empty if and only if $M$ has a non-empty intersection with $u_1^{(j)} \times \dots \times u_d^{(j)} = 1$.
            \item Delete tuples from $R_2, \dots, R_{d+1}$ to prepare for round $j+1$.
        \end{itemize}
    \end{enumerate}
    If $Q_d$ can be maintained in $O(|D|^{(d-1)/d - \epsilon})$ time per update, the total update time would be $O(|D|^{(d-1)/d+(\gamma+1)/d - \epsilon}) = O(|D|^{1+\gamma/d-\epsilon}) = O(n^{d+\gamma-d\epsilon})$, and the enumeration takes $O(1)$ per query, which thus solves OuMv$_k$ problem in time $O(n^{d+\gamma-\epsilon})$ and violates OuMv$_k$ Conjecture.
\end{proof}

\begin{lemma}
    Consider the star query
    \(
        Q_d'(x_{d-j+1}, \cdots, x_d) \gets R_1(x_1, x_2, \cdots, x_d), R_2(x_1), \cdots, \\ R_{k+1}(x_{k}),
    \)
    for any $0 \le d-j \le k \le d$. Unless OuMv$_k$ Conjecture failed, $Q_d'$ cannot be maintained in $O(|D|^{(d-1)/d-\epsilon})$ time while supporting $O(|D|^{(d-j)/d-\epsilon})$-delay enumeration for any $\epsilon > 0$.
\end{lemma}

\begin{proof}
    We adapt the construction from the proof for $Q_d$. In this scenario, enumeration requires outputting slices of the tensor $M$. For every output tuple $(t_{d-j+1}, \dots, t_d)$, a valid join result exists for the first $k$ dimensions. Consequently, for each candidate output, we only need to verify if the intersection of the remaining dimensions ($k+1$ to $d$) exists in $u_{k+1} \times \dots \times u_d$.   We can construct a hash index on the Cartesian product $u_{k+1} \times \dots \times u_d$ to perform this check in $O(1)$ time per entry. Computing this index takes $O(|D|^{(d-k)/d})$ time per round, and performing the distinct projection on $Q'_d$ to the remaining dimensions raises an additional $O(|D|^{(k+j-d)/d })$ delay. 
    
    If the amortized update time were $O(|D|^{(d-1)/d - \epsilon})$ and enumeration delay for enumerating $\pi_{x_{k+1}, \cdots, x_d} Q_d'$ were $O(|D|^{(d-j)/d-\epsilon}\times |D|^{(k+j-d)/d }) = O(|D|^{k/d}-\epsilon)$, since there are total $O(|D|^{(d-k)/d})$ output per round, the total running time would be $O(n^{\gamma+d-1-\epsilon'})$, which again violates the conjecture. 
\end{proof}

\subsection{Proof of Theorem~\ref{thm:dimension}}

We prove the theorem via a reduction from the hardness of the Star Queries ($Q_d$ and $Q_d'$) to an arbitrary query $Q$ with $\dim(Q) = d$.  Let $\V_c$ be the connected subset, with $\E_c$ the induced active relations, $\mathcal{N}$ the distinct neighbor set, and $\V_k(\mathcal{N})$ the key set that attains the maximum dimension (i.e., $\dim(\V_c) = d$). We construct an instance of $Q$ to simulate the hard query $Q_d$ (or $Q_d'$) as follows.

Let $M$ be the input tensor and $u_1, \dots, u_d$ be the input vectors for the OuMv$_k$ problem instance associated with $Q_d$. We map this instance to $Q$ through the following steps:

\begin{enumerate}[leftmargin=*]
    \item We map the $d$ dimensions to the attributes in $Q$.  If $\V_c \cap \bar{y} \neq \emptyset$ or $\V_k(\mathcal{N}) \subseteq \bar{y}$, we simply map the dimensions $x_1, \dots, x_d$ to the attributes in $\V_k(\mathcal{N})$.  Otherwise, we map the first $x_1, \cdots, x_k$ to the attributes in $\V_k(\mathcal{N})$, and the remaining dimensions $x_{k+1}, \dots, x_d$ to $(\bar{y} \cap \V_\star) \setminus \V_k(\mathcal{N})$.
    All attributes involved in this mapping are assigned the domain $[n]$.
    
    \item The attributes in the connected subset $\V_c$ serve as "glue" to reconstruct the tensor $M$. We assign their domain to be $[N]$, where $N = n^d$.
    Let $\V'$ be the set of all active attributes in the reduction: $\V' = \V_c \cup \V_k(\mathcal{N})$ or $\V' = \V_c \cup \V_k(\mathcal{N}) \cup (\bar{y} \cap \V_\star)$ depends on $\V_c$ and $\V_k(\mathcal{N})$. Any attributes $v \in \V \setminus \V'$ are fixed to a constant value $\perp$.

    We define the update sequence for $R(\bar{x})$ for all $\bar{x} \cap \V' \neq \emptyset$ as follows:
    
   \item \textbf{Tensor Encoding:} Consider relations where $|\bar{x} \cap \V'| \ge 2$. These relations form the core structure. We index the non-zero entries of the tensor $M$ with a unique identifier $ID \in [1, N]$. For each relation $R(\bar{x})$ with $\V_c \cap \bar{x} \neq \emptyset$, we project (without removing the duplicates) the tensor $M$ onto the attributes $\bar{x} \setminus \V_c$, and set those variables in $\V_c$ to the corresponding tensor entry $ID$.  If $R(\bar{x})$ contains \emph{only} variables from $\V_c$ (i.e., $\bar{x} \cap \V' \subseteq \V_c$), we simply insert tuples $(ID, \dots, ID)$ for all active IDs. This ensures the connectivity of $\V_c$ propagates the synchronization $ID$ across the query.  If $R(\bar{x})$ contains no variables from $\V_c$, we insert the distinct projection of the tensor $M$ for attributes $\bar{x} \cap \V'$ without constraining the query outputs.  

    \item \textbf{Neutralizing Unused Relations.} For any relation $R(\bar{x})$ with $\bar{x} \cap \V' = \emptyset$, we populate it with the single tuple $(\perp, \dots, \perp)$ at the beginning of the stream. This ensures these relations do not restrict the join result.

    \item \textbf{Vector Updates:} After (3) and (4), we start to update vectors.
    Consider relations where $\bar{x} \cap \V' = \{v\}$ for some $v \in \V_k(\mathcal{N})$. By the definition of the Distinct Neighbor Set, such relations exist and uniquely correspond to specific dimensions. We map the update stream of vector $u_i$ directly to this relation. Inserting a value $c$ into this relation corresponds to setting $u_i[c] = 1$.
    
\end{enumerate}

\paragraph{Correctness.} 
The construction ensures that a join result exists if and only if there is a tuple identifier $ID$ (representing a non-zero entry in $M$) such that all its coordinate values are present in the corresponding satellite relations (representing the vectors $u_i$ being 1). This exactly reconstructs the condition $M(x_1, \dots, x_d) \land \bigwedge u_i(x_i)$. The size of the database is dominated by $M$ and the vector updates, preserving the $O(|D|^{(d-1)/d-\epsilon})$ lower bound.

\section{Maintaining Star queries with Parameterized Complexity}
\label{sec:upper}
In Sections~\ref{sec:height} and \ref{sec:dimension}, we established that a query $Q$ with height $k$ and dimension $d$ cannot be maintained better than amortized $O\bigl(|D|^{1 - 1/\max\{k,d\} - \epsilon}\bigr)$ time. In this section, we provide a matching upper bound for the class of Star Queries. We introduce a data-dependent parameter, the \emph{Generalized H-index}, and design an algorithm with amortized update time $O(h(D)^{\dim(Q)-1})$. Since $h(D)$ is naturally bounded by $|D|^{1/\dim(Q)}$, this proves the optimality of our lower bound for these queries.

\subsection{Generalized H-index}
The classical H-index is a widely used metric that measures both the productivity and the citation impact of a scholar's publications.  A scholar has an H-index of $10$ if they have at least $10$ publications, each with at least $10$ citations, and they do not have $11$ publications with at least $11$ citations.  The same idea has been used in parameterized graph algorithms~\cite{10.1007/978-3-642-03367-4_25,EPPSTEIN201244}, where one partitions the graph into a dense part (induced by the vertices with many neighbors) and a sparse remainder.  However, the standard H-index is defined for binary relations such as graphs.  We therefore introduce a generalized version to measure the density of a relation in a star query.

\begin{definition}[Generalized H-index]
    Consider a star query $Q$ with center relation $R_1(x_1, \dots, x_d)$ and dimension $d = \dim(Q)$. For any attribute $x_i \in \{x_1, \dots, x_d\}$, let $\deg_{R_1}(v)$ denote the frequency of a value $v$ in the projection $\pi_{x_i} R_1$.  We define the H-index of a specific attribute $x_i$, denoted $h_i(D)$, as the largest integer such that there are at least $h_i(D)$ distinct values in $\pi_{x_i} R_1$ with frequency at least $h_i(D)^{d-1}$. Formally:
    \[
        h_i(D) = \max \left\{ k \in \mathbb{N} \;\middle|\; \left| \left\{ v \in \pi_{x_i} R_1 : \deg_{R_1}(v) \ge k^{d-1} \right\} \right| \ge k \right\}.
    \]
    The \emph{Generalized H-index} of the query $Q$ over database $D$ is defined as $h(D) = \max_{i=1}^d h_i(D)$.
\end{definition}

\begin{proposition}
    For any database instance $D$, the generalized H-index satisfies $h(D) \le |D|^{1/d}$.
\end{proposition}
\begin{proof}
    By definition, there exists an attribute $x_i$ with at least $h(D)$ distinct values, each appearing in at least $h(D)^{d-1}$ tuples in $R_1$. These sets of tuples are disjoint (grouped by $x_i$), so $|R_1| \ge h(D) \times h(D)^{d-1} = h(D)^d$. Thus, $h(D) \le |R_1|^{1/d} \le |D|^{1/d}$.
\end{proof}

\subsection{Algorithm Designs}

Our maintenance strategy relies on a \emph{Heavy-Light Decomposition} of the data space. Let $h(D)$ be the current generalized H-index. We classify values in the domain of each attribute $x_i$ as \emph{Heavy} if their degree exceeds $\tau_i = (h_i(D)+1)^{d-1}$, and \emph{Light} otherwise.

\paragraph{Partitioning}
For each attribute $x_i$, we define the set of heavy keys $K_i^H = \{v \in \pi_{x_i}R_1 \mid \deg_{R_1}(v) \ge \tau_i\}$.
We partition the center relation $R_1$ into $2^d$ logical partitions based on the heaviness of its attribute values. Each partition is indexed by a bit-vector $\mathbf{b} \in \{0,1\}^d$:
\[
    R_1^{(\mathbf{b})} = \left\{ t \in R_1 \;\middle|\; \forall i \in [1, d], \ 
    \begin{cases} 
        t[x_i] \in K_i^H & \text{if } \mathbf{b}[i] = 1 \text{ (Heavy)} \\
        t[x_i] \notin K_i^H & \text{if } \mathbf{b}[i] = 0 \text{ (Light)}
    \end{cases}
    \right\}
\]
Similarly, each satellite relation $R_i$ is partitioned into $R_i^{(1)}$ and $R_i^{(0)}$.  In addition, for $x_{k+1}$ to $x_{d}$, let $k+1 \le i \le d$ we conceptually maintain a vector $R_{i+1}^{(1)} \gets \pi_{x_i} R_1^{\mathbf{b}[i] = 1}$ for those heavy value.  

\paragraph{View Maintenance}
We maintain the query $Q$ as a union of $2^d$ subqueries $Q^{(\mathbf{b})}$, corresponding to $R_1^{(\mathbf{b})}$. For a fixed $\mathbf{b}$, let $I_H = \{i \mid \mathbf{b}[i] = 1\}$ be the indices of heavy attributes. We maintain:

\begin{enumerate}[leftmargin=*]
    \item \textbf{Satellite Projections:} For every $R_i^{(\mathbf{b}[i])}$, we maintain its projection: $V_{R_i}^{(\mathbf{b}[i])} = \pi_{x_{i-1}} R_i^{(\mathbf{b}[i])}$.
    
    \item \textbf{Light-Filtered Core ($V_S$):} For the center relation $R_1^{(\mathbf{b})}$, we actively filter it by the \emph{light} satellites. We maintain the semi-join view:
    \(
        V_S^{(\mathbf{b})} = R_1^{(\mathbf{b})} \ltimes \left( \mathop{\ltimes}_{i \notin I_H} V_{R_{i+1}}^{(0)} \right).
    \)
    
    Since light keys have a low degree, updates to light satellites are cheap to propagate.
    
    \item \textbf{Heavy Core ($V_C$):} We further project the light-filtered core onto the \emph{heavy} attributes and filter by the Cartesian product of \emph{heavy} satellites:
    \(
        V_C^{(\mathbf{b})} = \left( \pi_{\{x_i\}_{i \in I_H}} V_S^{(\mathbf{b})} \right) \ltimes \left( \mathop{\times}_{i \in I_H} V_{R_{i+1}}^{(1)}\right).
    \)
    
    $V_C$ contains only tuples composed entirely of heavy keys that have survived all satellite filters.
\end{enumerate}

\paragraph{Enumeration}
To enumerate the result of $Q$, we iterate through all $\mathbf{b} \in \{0,1\}^d$. For each $\mathbf{b}$, we
1. Enumerate tuples $t_{core} \in V_C^{(\mathbf{b})}$.
2. Use $t_{core}$ to probe the index of $V_S^{(\mathbf{b})}$ to find matching tuples $t_{full} \in R_1$ that survived the light filters.
3. For each $t_{full}$, extend it to the final result by joining with the satellites (using standard constant delay lookups).
This process guarantees $O(1)$-delay enumeration because the query is free-connex~\cite{wang2023change,idris17:_dynam}.  We then perform a distinct project on the output attributes $\bar{y}$, which raises an additional $O(|D|^{(d-j)/d})$ delay.

\paragraph{Rebalance} So far, we have assumed that the heavy/light partition is fixed, and $h(D)$ is also fixed.  In the dynamic setting, however, the classification of a key as heavy or light, and the generalized H-index $h(D)$, may change over time as updates arrive.  We handle these changes with a standard rebalancing technique similar to the textbook method for dynamic tables~\cite{cormen09:_introd}.

\begin{enumerate}[leftmargin=*]
    \item \textbf{Key Reclassification:} We promote a key $x_i$ to heavy only when its degree exceeds $2\tau_i$ and demote it only when it drops below $\frac{1}{2}\tau_i$. This ensures that between any two status changes of a key, $\Omega(h_i(D)^{d-1})$ updates involving that key must have occurred.  
    
    To perform such a rebalance, we first delete all associated tuples from $R_1$, and then all associated tuples from $R_{i+1}$, and reinsert them back to $R_{i+1}$ and then $R_1$ with the updated label.  These updates pay for the $O(h_i(D)^{d-1})$ work required to move tuples between partitions, as the frequency of such a key would be $O(h_i(D)^{d-1})$, and incur an amortized $O(1)$ cost per update.
    
    \item \textbf{Global Parameter Update:} If the number of heavy keys in any dimension exceeds $2h(D)$, or if $h(D)$ decreases by half compared with the last parameter update, we rebuild the data structures with a new $h(D)$ that can be tracked in real time. 
    
    When a global rebalance happens, at most $O(h(D))$ keys need to be moved between the heavy and light sets, and in the worst case, it costs $O(|D|) = O(h(D)^d)$. However, to reach a state requiring a rebuild, $\Omega(h(D) \cdot h(D)^{d-1}) = \Omega(h(D)^d)$ updates must have occurred, e.g., at least $h(D)$ keys are moved from light to heavy, with every key's frequency changes from at most $h(D)^{d-1}$ in previous global rebalance to $2h(D)^{d-1}$; or at least $0.5h(D)$ heavy keys are moved from heavy to light, with frequency changes from at least $h(D)^{d-1}$ in previous global rebalance to $0.5h(D)^{d-1}$.   Distributing the cost over these updates yields an amortized cost of $O(1)$.
\end{enumerate}

\subsection{Complexity Analysis}

It is not hard to see that the space requirement for the algorithm is linear to $O(|D|)$; therefore, to complete the proof of Theorem~\ref{thm:upper}, it suffices to analyze the amortized update complexity.  The update cost consists of three parts: maintenance cost, rebalancing cost, and the cost of tracking $h(D)$. In these three parts, we have also demonstrated that the rebalancing step incurs only an amortized constant cost per update; we now focus on the remaining two parts.

\paragraph{Maintenance cost} 
We analyze the cost of a single tuple update, assuming $h(D)$ is constant.
\begin{itemize}[leftmargin=*]
    \item \textbf{Update to Center $R_1$:} Inserting a tuple $t$ into $R_1$ affects exactly one partition $R_1^{(\mathbf{b})}$. Checking the semi-joins for $V_S$ and updating $V_C$ takes $O(1)$ time per update with proper hash indices.
    \item \textbf{Update to Light Satellite $R_i^{(0)}$:} A tuple inserted here has a join key $k$ with $\deg_{R_1}(k) < \tau_i$. This update may trigger insertions into $V_S^{(\mathbf{b})}$. Since the degree is bounded by $\tau_i = O(h_i(D)^{d-1})$, at most $O(h_i(D)^{d-1})$ tuples in $R_1$ are activated. Propagating these to $V_C$ takes an additional constant time per $R_1$ tuple, making the total cost $O(h_i(D)^{d-1})$.
    \item \textbf{Update to Heavy Satellite $R_i^{(1)}$:} The join key $k$ is heavy. This does not affect $V_S$ but may update the Cartesian product $ \mathop{\times}_{i \in I_H} V_{R_{i+1}}^{(1)}$. Since there are at most $h_i(D)$ heavy values per dimension, the number of the affected tuples in the Cartesian product is effectively bounded by $O(h(D)^{|I_H|-1})$, thus the update on $V_C$ is upper bounded by $O(h(D)^{|I_H|-1}) \le O(h(D)^{d-1})$.
\end{itemize}
Therefore, with a fixed heavy/light classification, the update cost is $O(h(D)^{d-1})$.

\paragraph{Tracking $h(D)$} One additional requirement for the maintenance algorithm is to track the change of $h_i(D)$ for all dimensions $i$ and thus $h(D)$.  To do that, for every dimension $i$, we maintain two group-by-count queries $V_{\mathsf{cnt}_i} = \pi_{x_i, \mathsf{count}(*) \text{ as deg}} R_1$ and $V_{\mathsf{feq}_i} = \pi_{\mathsf{deg}, \mathsf{count}(*) \text{ as feq}} V_{\mathsf{cnt}_i}$.  These two queries can be maintained in $O(1)$ time~\cite{chirkova2012materialized} under single-tuple update on $R_1$, and we can obtain the real-time $h_i$ by maintaining the result 
\(h_i = \max \left\{\mathsf{cnt} \left| \sum_{i = \mathsf{cnt}}^{|D|^{1/\dim(Q)}} \left( V_{\mathsf{feq}_i}[\mathsf{cnt}].\mathsf{feq}\right) \ge \mathsf{cnt}^{\dim(Q) - 1}  \right.\right\},\)
which can also be maintained in $O(1)$ time by keeping a frequency histogram.   

Combining all parts, the amortized update time per tuple update is $O\bigl(h(D)^{\dim(Q)-1}\bigr)$.  Since we always have $h(D) \le |D|^{1/\dim(Q)}$, this matches the lower bound from Theorem~\ref{thm:dimension} for all star queries, and thus proves tightness for this class.

\section{Additional Discussions and Future Works}
\label{sec:con}

In this work, we introduce two structural parameters, \emph{height} and \emph{dimension}, to establish parameterized lower bounds for maintaining conjunctive queries under updates.  We shall note that \emph{height} assumes the combinatorial $k$-clique conjecture, which only rules out the existence of a combinatorial algorithm.  Although the notion of \emph{combinatorial algorithm} is not formally defined, it is widely used across different research communities, including in database theory \cite{10.1145/3584372.3588666, fan2023fine, conjunctiveLB}, due to the practical efficiency, elegance, and generalizability of combinatorial algorithms \cite{abboud2024faster}. This also immediately raises an interesting research problem for both upper-bound and lower-bound designs: can we apply fast matrix multiplication to accelerate dynamic query evaluation?


Another critical question is whether a maintenance algorithm can be designed to match these lower bounds.  To complement our lower bounds, we investigated the maintenance of \emph{star queries}, a central subclass of analytical workloads. We established a parameterized upper bound based on the \emph{Generalized H-index} $H(D)$. It generalizes prior trade-off algorithms~\cite{kara2020trade,kara2020maintaining} to an instance-dependent setting, demonstrating that efficiency often depends on data skew rather than worst-case size.  The generalization of such techniques would benefit both theory and system research.  \rwa{On the other hand, it remains an open question to achieve a trade-off between the enumeration delay and maintenance cost, especially for queries with $\omega(Q) > 1$. }

Furthermore, closing the gap for general queries remains a significant challenge. Specifically, new algorithmic designs are required to handle \emph{Loomis-Whitney}-like queries~\cite{10.1145/3695837, kara2020maintaining}, which contain nested join keys that complicate the new maintenance techniques for star queries. Future avenues for research include extending both upper bounds and lower bounds to general scenarios, such as comparisons~\cite{idris2018conjunctive, comparison}, unions~\cite{carmeli2019enumeration}, and queries over specific update patterns~\cite{10.1145/3695837, pods25}.

\appendix
\section{Proof of Lemma~\ref{lem:height-hierarchical}}

\textbf{If Direction: } Let $Q = (\V, \E, \bar{y})$ be a height $1$ query.  If $Q$ is not q-hierarchical, then (1) there exists $x_1, x_2 \in \V$ such that $\E[x_1] \cap \E[x_2] \neq \emptyset$ and $\E[x_1] \setminus \E[x_2] \neq \emptyset$, $\E[x_2] \setminus \E[x_1] \neq \emptyset$; or (2) $x_1 \in \bar{y}$, $x_2 \notin \bar{y}$, while $\E[x_1] \subsetneq \E[x_2]$.  However, for (1), let $e_1 \in \E[x_1] \cap \E[x_2]$, $e_2 \in \E[x_1] \setminus \E[x_2]$, $e_3 \in \E[x_2] \setminus \E[x_1]$, $P \leftarrow \{x_1, x_2\}$ forms a length $3$ chordless path with $P_e \gets \{e_2, e_1, e_3\}$ be one of its edge sequence, contradicting that $Q$ is height $1$; and for (2), let $e_1 \in \E[x_1] \cap \E[x_2]$, $e_2 \in \E[x_1] \setminus \E[x_2]$, then $P\gets \{x_2\}$ is a length $2$ q-chordless path, with $P_e \gets \{e_1, e_2\}$ be an edge sequence with $e_1$ contains output attributes $x_2$, which also contradicts to $Q$ is height $1$.  Therefore, the assumption that $Q$ is not q-hierarchical is incorrect.  

\textbf{Only-if Direction: } Similarly, let $Q$ be a q-hierarchical query; we cannot find a length $\ge 3$ chordless path, otherwise we can find a pair of attributes $x_1, x_2$ such that $\E[x_1] \cap \E[x_2] \neq \emptyset$ and $\E[x_1] \setminus \E[x_2] \neq \emptyset$, $\E[x_2] \setminus \E[x_1] \neq \emptyset$.  Similarly, we cannot find a length $\ge 2$ q-chordless path; otherwise, we can find attributes $x_1, x_2$ such that $x_1 \in \bar{y}$, $x_2 \notin \bar{y}$, while $\E[x_1] \subsetneq \E[x_2]$.  Both cases violate the fact that $Q$ is q-hierarchical; therefore, the query must have height $1$.

\section{Proof of Lemma~\ref{thm:tree-height}}
\label{app:heights}



To prove Lemma~\ref{thm:tree-height}, we first demonstrate the recursive properties of the height definition.

\begin{lemma}
\label{lem:3.5}
    A free-connex query $Q$ is a height $k$ query for $k > 1$ if and only if its residual query $Q'$ is a height $k-1$ query.
\end{lemma}

\begin{proof}
    The lemma is immediate for the base case when $k = 2$, i.e., the weak-q-hierarchical queries, as their residual queries are q-hierarchical.  We now consider the general case $k > 2$.  By definition of residual query and ears, given a free-connex join tree $\T$, each leaf node corresponds to an ear, and its parent node is the anchor of that ear. Therefore, we deduce two important facts: (1) if $Q$ has a join tree of height $k$, then its residual query $Q'$ has a join tree of height at most $k-1$. This is because removing all leaf nodes (ears) from $\T$ yields a join tree with height $k-1$. If necessary, further removing any remaining ears from the resulting tree can only decrease the height further. (2) Consequently, if the residual query $Q'$ admits a free-connex join tree $\T'$ of height $k-1$, then we can construct a join tree for $Q$ of height at most $k$ by attaching each removed ear $R$ as a child of its anchor $R'$ in $\T'$. The resulting structure remains a valid free-connex join tree, satisfying the Connect, Cover, and Connex conditions without violating the Guard or Above conditions. Since adding these ears introduces only new leaves, the tree's height increases by 1; therefore, the reconstructed join tree for $Q$ has height at most $k$.

    We first prove the if direction. Suppose the residual query $Q'$ is a height $k-1$ query. Assume $Q$ admits a free-connex join tree $\T$ of height $\leq k-1$. Removing all ears from $\T$ yields a join tree for $Q'$ of height at most $k-2$, by observation (1) above. This contradicts the assumption that $Q'$ is a height $k-1$ query, since $Q'$ would then have a join tree of height smaller than $k-1$. Therefore, $Q$ cannot have any free-connex join tree of height less than $k$. By observation (2), on the other hand, we know how to construct a height-$k$ join tree for $Q$. Taken together, the minimum possible join tree height for $Q$ is $k$. In other words, $Q$ is a height $k$ query.

    For the “only-if” direction, assume $Q$ is a height $k$ query. By definition, every free-connex join tree of $Q$ has height at least $k$, and there is some join tree $\T$ with height $k$. Remove all the leaf nodes from this $\T$ to obtain a join tree $\T'$ with height $k-1$, and further removing ears from $\T'$ will only decrease the height, making $ Q'$'s height at most $k-1$. Now, suppose for contradiction that $Q'$ is not a height $k-1$ query. Then there exists a free-connex join tree for $Q'$ of height at most $k-2$. By attaching the previously removed ears to this shorter join tree using the construction from observation (2), we could obtain a join tree for $Q$ of height at most $k-1$. This would contradict the assumption that $Q$'s height is $k$. Therefore, $Q'$ must have minimal join tree height $k-1$, meaning $Q'$ is a height $k-1$ query. 
\end{proof}

Since every height $1$ query is q-hierarchical and with a join tree of height $1$ (Lemma~\ref{lem:height-hierarchical}), and let $Q'$ be a residual query of $Q$ by removing all ears from the query, and $T'$ be a height $k'$ free-connex join tree for $Q'$, a height $k'+1$ free-connex join tree for $Q$ can be constructed by adding each ear node as a leaf child under its anchor node and increasing the height of tree by $1$.  Therefore, we can construct a height $k$ free-connex join tree for a height $k$ query $Q$ by the following steps: (1) we recursively removing all ear nodes and get residual queries $Q_2, \cdots, Q_k$, where $Q_k$ has a height of $1$; (2) we find the height $1$ free-connex join tree $T_k$ for $Q_k$; (3) for $Q_{k-1}, \cdots, Q_2, Q$, we add all ears back to $T_k$ and obtain $T_{k-1}, \cdots, T_{2}, T$ with height $2, \cdots, k-1, k$. 

In addition, to show that $Q$ admits no height $k-1$ free-connex join tree, we prove that every chordless path must be realized within any free-connex join tree:

\begin{lemma}
\label{lem:treepath}
    Given a CQ $Q = (\V, \E, \bar{y})$, let $P \gets \{v_1, \cdots, v_{k-1}\}$ be a length $k$ chordless path on $Q$'s relational hypergraph $\H$, for any join tree $\T$ of $Q$, let $n_1$ be a node with $v_1 \in n_1$ while $v_i \notin n_1$ for all $i \in [2, k-1]$, and $n_l$ be a node with $v_{k-1} \in n_l$ while $v_i \notin n_l$ for all $i \in [1, k-2]$, then the path $P_\T$ between $n_1$ and $n_l$ on the join tree includes all attributes in $P$.  In addition, such a path has a length of at least $k$.  
\end{lemma}

\begin{proof}[Proof of Lemma~\ref{lem:treepath}]
    We start by considering a length-2 path $v_1, v_2$.  Because of the definition, we know that there exists $e_1$, $e_2$, $e_3$ such that $v_1 \in e_1$ and $v_2 \notin e_1$, $v_1, v_2 \in e_2$, and $v_2 \in e_3$ and $v_1 \notin e_3$.  Because the Connect property of a join tree, let $T_1$ be the subtree induced by $v_1$, and $T_2$ be induced by $v_2$, since $e_2$ belongs to both $T_1$ and $T_2$, therefore, all nodes with $v_1$ or $v_2$ are also connected on $\T$.  We let $e_1$ be $n_1$ and $e_3$ be $n_2$; there must exist a node $n'$ on the path between $n_1$ and $n_2$ with $v_1, v_2 \in n'$, otherwise the path contains an addition node $v_3$ in between $v_1$ and $v_2$, combining the connect condition with $e_2$ creates a cycle, which violates the fact that $\T$ is a tree.

    On the other hand, assuming for a length $k-1$ chordless path, the lemma holds, then for a length $k$ chordless path, we can adopt a similar idea as for the length-2 case to show the connectness of $v_1, \cdots, v_{k-1}$-induced subtree.  Therefore, let $n_1$ and $n_{k}$ be two endpoints, the path between $n_1$ and $n_k$ must cover $v_1, \cdots, v_{k-1}$, otherwise a cycle would be detected on the tree.  By applying the induction, we can prove the lemma.

    Because there is no shortcut, covering $v_1, \cdots, v_{k-1}$ requires at least $k-2$ nodes, and there are two endpoints for $v_1$ and $v_{k-1}$, making such a path have a length of at least $k$.
\end{proof}

\begin{lemma}
\label{lem:heightlb}
    Given a CQ $Q=(\V, \E, \bar{y})$, let $P \gets \{v_1, \cdots, v_{k-1}\}$ be a length $k$ chordless path on $Q$'s relational hypergraph $\H$, for any join tree $\T$ of $Q$, $\T$ has a height at least $\lceil\frac{k}{2}\rceil$.
\end{lemma}

\begin{proof}[Proof of Lemma~\ref{lem:heightlb}]
We can find a corresponding path $P_\T$ on any join tree $\T$ for the chordless path $P$, which has a length of at least $k$.  Based on the definition of height, we can minimize the height of the tree by placing the two endpoints to be leaf nodes of $\T$ and requiring the path to pass the root node.  Assuming the distance between the two leaf nodes and the root node is balanced, the height would be at least $\lceil\frac{k}{2}\rceil$.
\end{proof}

Therefore, for a Boolean or full acyclic query of height $k$, all join trees for that query have a height of at least $k$.  When considering non-full and non-Boolean free-connex queries, we can then demonstrate the connection between a q-chordless path and the height of a join tree.
\begin{lemma}
\label{lem:heightqlb}
    Given a free-connex $Q=(\V, \E, \bar{y})$, let $P \gets \{v_1, \cdots, v_{k-1}\}$ be a length $k$ q-chordless path on $Q$'s relational hypergraph $\H$, for any free-connex join tree $\T$ of $Q$, $\T$ has a height at least $k$.
\end{lemma}

\begin{proof}[Proof of Lemma~\ref{lem:heightqlb}]
\citet{pods25} shows that for a height $2$ query with $\bar{y}$ to be non-empty or non-full, it is possible there is no length $3$ or length $4$ path, but contains a length $2$ q-chordless path $P = \{x_1\}$ with $x_1 \notin \bar{y}$, $R_1(\bar{x}_1), R_2(\bar{x}_2)$ be two relations with $x_1 \in \bar{x}_1 \cap \bar{x}_2$, and $\exists x \in \bar{x}_1 \cap \bar{y}$.  In that case, $R_1$ must be the parent node of $R_2$ on any free-connex join tree, otherwise violating the Connex Property of the join tree.  On the other hand, if for a height $k$ query $Q$, its residual query $Q'$ satisfies $\H[Q']$ contains such a length $k-1$ q-chordless path, by definition of q-chordless path, let $R_{k-1}(\bar{x}_{k-1})$ be the relation with only $x_{k-2}$ in the residual query $Q'$, since $R_{k-1}$ is not an ear in $Q$ and $R_{k-1}$ contains only non-output attributes, there exists an ear $R_k(\bar{x}_k)$ in $Q$ with $R_{k-1}$ be its anchor, and $\bar{x}_k \cap \bar{x}_{k-1} \cap P = \emptyset$ and $\bar{x}_k \cap \bar{x}_{k-1} \neq \emptyset$.  Therefore, let $x_{k} \in \bar{x}_k \cap \bar{x}_{k-1}$ be any join attributes between $R_{k-1}$ and $R_{k}$, we can find a length $k$ q-chordless path with $P' \gets P \cup \{x_k\}$, which completes the proof.    
\end{proof}

Combining the two lemmas, we can show that a height $k$ query cannot admit a height $k-1$ free-connex join tree.
 




\section{Proof of Lemma~\ref{lem:hierarchical}}

For a q-hierarchical query $Q$, we can find that there is no $\V_c$ with $R_1, R_2 \in \E_c$ such that $\bar{x}_1 \cap \bar{x}_2 \notin \V_c$; otherwise the query won't be q-hierarchical.   

Let $Q = (\V, \E, \bar{y})$ be a CQ, if any $\V_c \in \V$ has $\dim(\V_c) > 1$, consider the $\mathcal{N}$ and $\V_k(\mathcal{N})$ that would maximize the dimension of $\V_c$, if (1) $\V_c \cap \bar{y} \neq \emptyset$ or $\V_k(\N) \subseteq \bar{y}$, there exists two $R', R'' \in \mathcal{N}$, such that $x_1 \in \bar{x}'$ and $x_2 \in \bar{x}''$.  If there exists $R \in \E_c$ with $x_1, x_2 \in \bar{x}$, then $\E[x_1] \cap \E[x_2] \neq \emptyset$, while $\E[x_1] \nsubseteq \E[x_2]$ and $\E[x_2] \nsubseteq \E[x_1]$, making the query non-q-hierarchical; otherwise, let $R_1, R_2 \in \E_c$ with $x_1 \in \bar{x}_1$ and $x_2 \in \bar{x}_2$, there must be $x_3 \in \V_c \cap \bar{x}_1 \cap \bar{x}_2$, and $x_1/x_2$ and $x_3$ will violate the q-hierarchical condition.   

On the other hand, if (2) $\V_c$ are all non-output attributes, and $\V_k(\N) \setminus \bar{y} \neq \emptyset$, it is possible that there is a unique output attribute $x \in \V_\star$ that will contribute to the dimension.  In that case, if $\E_c$ contains only one relation, we can find a non-output attribute $x' \in \V_k(\N)$, we know $\E[x] \subsetneq \E[x']$, and $x \in \bar{y}$, but $x' \notin \bar{y}$, making the query non-q-hierarchical; otherwise, if $\E_c$ contains more than one relation, we can find a non-output attribute $x' \in \V_c$, and $\E[x] \subsetneq \E[x']$, and $x \in \bar{y}$, but $x' \notin \bar{y}$, making the query non-q-hierarchical.

For the only-if direction, since the query contains only dimension-$1$ $\V_c$, which holds for $\V_c = \emptyset$ and all $\E_c = \{R\}$, we can show that if, for all $x_1, x_2$ in $\V$, $\E[x_1] \cap \E[x_2] = \emptyset$ always holds, then all relations in the query are disconnected, and the query is q-hierarchical.  If $\E[x_1] \cap \E[x_2] \neq \emptyset$, then $\E[x_1] \subseteq \E[x_2]$ or $\E[x_2] \subseteq \E[x_1]$, otherwise let $R \in \E[x_1] \cap \E[x_2]$, we can show that $\dim(\{\emptyset\}) \ge 2$ for $\E_c = \{R\}$, with a $\mathcal{N} \gets \{R', R''\}$ with $R' \in \E[x_1] \setminus \E[x_2]$ and $R'' \in \E[x_2] \setminus \E[x_1]$ and use $x_1$ and $x_2$ be their distinct keys.  In addition, if there exists $\E[x_1] \subsetneq \E[x_2]$, and $x_1 \in \bar{y}$, then it must have $x_2 \in \bar{y}$, otherwise for any $R$ with both $x_1, x_2 \in \bar{x}$, when selecting $\{R'\}$ to be the distinct neighbour set $\mathcal{N}$ for $\V_c = \{\emptyset\}$ and $\E_c = \{R\}$, with $x_2 \in \bar{x}'$, $\V_k(\mathcal{N}) \nsubseteq \bar{y}$ because of $x_2$, and $x_1$ is a live output attribute, making the dimension of $\V_c = \{\emptyset\}$ and $\E_c = \{R\}$ to be at least $2$.  Therefore, if the dimension of $Q$ is $1$, then the query is q-hierarchical.

\section{Additional Running Examples}
\label{app:example}

\begin{example}
\begin{figure}[H]
\resizebox{0.8\linewidth}{!}{
\begin{tikzpicture}[font=\Huge\bfseries,
    solid_node/.style={
        circle, 
        draw=black, 
        fill=black, 
        inner sep=0pt, 
        minimum size=7pt,        
        text=black,
        label distance=6pt 
    },
    hollow_node/.style={
        circle, 
        draw=black, 
        fill=white, 
        inner sep=0pt, 
        minimum size=7pt,        
        label distance=6pt
    },
    main_ellipse/.style={
        ellipse,
        draw,
        thick,
        fill opacity=0.15,
        inner sep=0pt,
        minimum height=2.4cm, 
    },
    branch_ellipse/.style={
        ellipse,
        draw,
        thick,
        fill opacity=0.15,
        inner sep=3pt,
        minimum height=1.6cm, 
    },
    rel_label/.style={
        fill=white,
        fill opacity=0.85,
        text opacity=1,
        inner sep=2pt,
        rounded corners=2pt
    },
    unary_ellipse/.style={
        circle,
        draw,
        thick,
        fill opacity=0.15,
        inner sep=18pt, 
    }
]

    
    \coordinate (C1) at (0, 6);      
    \coordinate (C3) at (-5, -3);    
    \coordinate (C2) at (5, -3);     

    \node[hollow_node, label=90:$x_1$] (x1) at (C1) {};
    \node[hollow_node, label=225:$x_3$] (x3) at (C3) {};
    \node[hollow_node, label=315:$x_2$] (x2) at (C2) {};


    \node[solid_node, label=45:$x_4$] (x4) at ($(x1)!0.33!(x2)$) {};
    \node[solid_node, label=45:$x_5$] (x5) at ($(x1)!0.66!(x2)$) {};

    \node[solid_node, label=135:$x_{15}$] (x15) at ($(x1)!0.33!(x3)$) {};
    \node[hollow_node, label=135:$x_{14}$] (x14) at ($(x1)!0.66!(x3)$) {};

    \node[solid_node, label=270:$x_{12}$] (x12) at ($(x3)!0.33!(x2)$) {};
    \node[solid_node, label=270:$x_{13}$] (x13) at ($(x3)!0.66!(x2)$) {};


    \node[hollow_node, label=90:$x_6$] (x6) at ($(x4)+(3.0, 0)$) {};
    \node[hollow_node, label=90:$x_8$] (x8) at ($(x6)+(3.0, 0)$) {};
    \node[hollow_node, label=90:$x_{10}$] (x10) at ($(x8)+(3.0, 0)$) {};

    \node[hollow_node, label=270:$x_7$] (x7) at ($(x5)+(3.0, -1.0)$) {};
    \node[hollow_node, label=270:$x_9$] (x9) at ($(x7)+(3.0, -1.0)$) {};
    \node[hollow_node, label=270:$x_{11}$] (x11) at ($(x9)+(3.0, -1.0)$) {};

    \node[solid_node, label=0:$x_{17}$] (x17) at ($(x13)+(-0.5, -2.0)$) {};
    \node[hollow_node, label=0:$x_{18}$] (x18) at ($(x17)+(-0.5, -2.0)$) {};
    \node[hollow_node, label=0:$x_{19}$] (x19) at ($(x18)+(-0.5, -2.0)$) {};

    \node[hollow_node, label=90:$x_{16}$] (x16) at ($(x14)+(-3.0, 0.5)$) {};


        
        \node[main_ellipse, fill=red, rotate fit=-60, fit=(x1) (x2), 
              label={[text=red!80!black, rel_label]0:$R_1$}] {};
        
        \node[main_ellipse, fill=blue, fit=(x3) (x2), 
              label={[text=blue!80!black, rel_label]270:$R_2$}] {};
        
        \node[main_ellipse, fill=green!80!black, rotate fit=60, fit=(x1) (x3), 
              label={[text=green!60!black, rel_label]180:$R_3$}] {};


        \node[branch_ellipse, fill=orange, rotate fit=0, fit=(x4) (x6), 
              label={[text=orange!90!black, rel_label]90:$R_4$}] {};
        \node[branch_ellipse, fill=violet, rotate fit=0, fit=(x6) (x8), 
              label={[text=violet, rel_label]90:$R_5$}] {};
        \node[branch_ellipse, fill=cyan, rotate fit=0, fit=(x8) (x10), 
              label={[text=cyan!80!black, rel_label]90:$R_6$}] {};

        \node[unary_ellipse, fill=yellow!90!black, fit=(x10), 
              label={[text=yellow!60!black, rel_label]0:$R_{13}$}] {};

        \node[branch_ellipse, fill=magenta, rotate fit=-18, fit=(x5) (x7), 
              label={[text=magenta!80!black, rel_label]90:$R_7$}] {};
        \node[branch_ellipse, fill=brown, rotate fit=-18, fit=(x7) (x9), 
              label={[text=brown, rel_label]90:$R_8$}] {};
        \node[branch_ellipse, fill=teal, rotate fit=-18, fit=(x9) (x11), 
              label={[text=teal, rel_label]90:$R_9$}] {};

        \node[unary_ellipse, fill=blue!50!cyan, fit=(x11), 
              label={[text=blue!50!cyan, rel_label]0:$R_{14}$}] {};

        \node[branch_ellipse, fill=olive, rotate fit=76, fit=(x13) (x18), 
              label={[text=olive, rel_label]270:$R_{10}$}] {};
        \node[branch_ellipse, fill=lime!80!black, rotate fit=76, fit=(x18) (x19), 
              label={[text=lime!80!black, rel_label]180:$R_{11}$}] {};

        \node[branch_ellipse, fill=purple, rotate fit=-10, fit=(x14) (x16), 
              label={[text=purple, rel_label]270:$R_{12}$}] {};


\end{tikzpicture}
}
\caption{Hypergraph for $Q$, where solid dots represent the output attributes.}
\label{fig:graph2}
\end{figure}
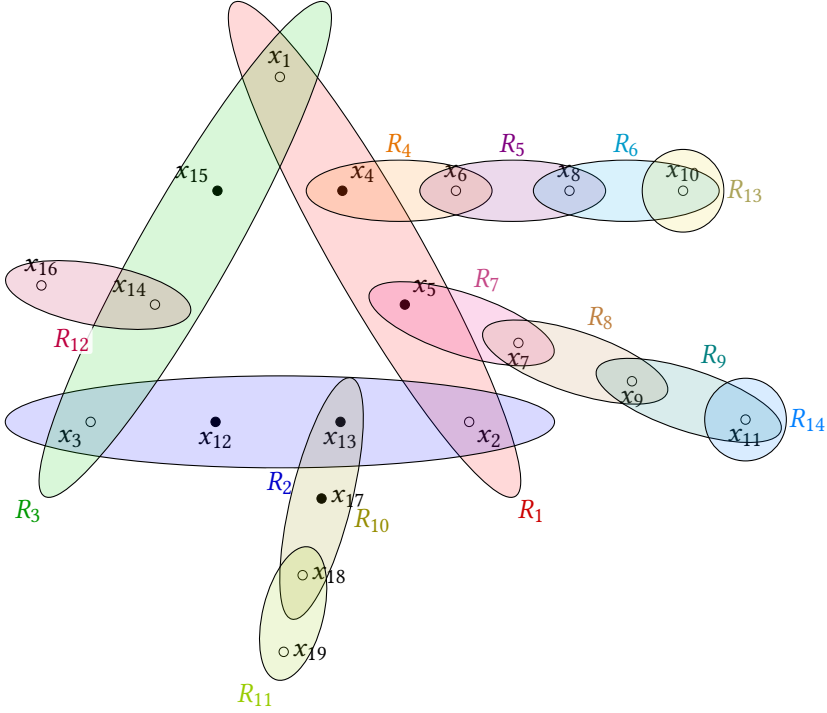
Consider the query:
\[
    \begin{aligned}
    Q(x_4, x_5, x_{12}, x_{13}, x_{15}, x_{17}) \gets & R_1(x_1, x_4, x_5, x_2), R_2(x_3, x_{12}, x_{13}, x_{2}), R_3(x_1, x_{15}, x_{14}, x_3),  \\ & R_4(x_4, x_{6}), R_5(x_{6}, x_8), R_6(x_8, x_{10}), R_7(x_5, x_7), R_8(x_7, x_9), R_9(x_9, x_{11}), 
    \\ &  R_{10}(x_{13}, x_{17}, x_{18}), R_{11}(x_{18}, x_{19}), R_{12}(x_{14}, x_{16}), R_{13}(x_{10}), R_{14}(x_{11})
    \end{aligned}
\]
Figure~\ref{fig:graph2} illustrates its relational hypergraph.  The query is a height $5$ query, with the length-$9$ chordless path $x_{11}, x_9, x_7, x_5, x_4, x_6, x_8, x_{10}$, or a length-$5$ q-chordless path $x_4, x_6, x_8, x_{10}$.  However, $x_{11}, x_9, x_7, x_5, x_2, x_3, x_1, x_4, x_6, x_8, x_{10}$ is not a chordless path due to the shortcut between $x_4$ and $x_5$.  

When considering the dimension, the query is a dimension $6$ query, with $\V_c = \{x_1, x_2, x_3\}$, $\E_c = \{R_1, R_2, R_3\}$, $\mathcal{N} = \{R_{12}, R_{10}, R_{4}, R_{7}\}$, $\V_k(N) = \{x_4, x_5, x_{13}, x_{14}\}$ and $\V_\star = \{x_4, x_5, x_{12}, x_{13}, x_{14},  x_{15}\}$.  All vertices in $\V_\star$ contribute to the dimension, as $\V_\star \setminus \V_k(N) = \{x_{12}, x_{15}\} \subseteq \bar{y}$ and $\V_c \cap \bar{y} = \emptyset$.  We cannot obtain a larger dimension by setting $\V_c = \{x_1, x_2, x_3, x_{13}\}$, $\E_c = \{R_1, R_2, R_3, R_{10}\}$ and $\V_\star = \{x_4, x_5, x_{12}, x_{14},  x_{15}, x_{17}, x_{18}\}$. Although the new $\V_\star$ contains more attributes, since $x_{13} \in \bar{y}$, we can only select $\V_k(\mathcal{N}) = \{x_4, x_5, x_{14}, x_{18}\}$ with $\mathcal{N} = \{R_4, R_7, R_{12}, R_{11}\}$ to count toward the dimension and $\dim(\V_c) = 4$ in this case.

When considering static database, the query can be evaluated in $O(|D|^{1.5}+|Q(D)|)$ by replacing $\{R_1, R_2, R_3\}$ with a bag  $B_1(x_1, x_2, x_3, x_4, x_5, x_{12}, x_{13}, x_{14}, x_{15}) \gets R_1, R_2, R_3$ and thus transforming the query into free-connex.  The running time suggests for this query, the width $\omega(Q) \le 1.5$, and the lower bound for dynamic evaluation from the width would be $\Omega(|D|^{\omega(Q)-1}) = \Omega(|D|^{0.5})$, which is dominated by the bound from the dimension of the query.   
\end{example}

\begin{example}[Running Example for Theorem~\ref{thm:dimension}]
\label{ex:4}
    Consider $Q_1$ from Example~\ref{ex:1}, we now $\V_c = \{x_5\}$ with $\E_c = \{R_3, R_7\}$ to simulate an OuMv$_4$ problem.  We select $\mathcal{N} \gets \{R_2, R_4, R_{10}, R_8\}$ with the distinct join keys $x_3, x_4, x_9, x_6$ (See Figure~\ref{fig:bipartite} for the illustration).   Given an OuMv$_4$ instance, let $M(a_1, a_2, a_3, a_4, l)$ be the sparse form of the tensor, with $l$ being the tuple identifier.  We create the update sequence to be: (1) We insert a special tuple $(\perp)$ into $R_6$ and $(\perp, \perp)$ into $R_1$ and $R_5$; (2) We insert all tuples in $\pi_{a_1, a_2, l, a_3} M$ into $R_3$, with $\{a_1 \rightarrow x_3, a_2 \rightarrow x_4, l \rightarrow x_5, a_3 \rightarrow x_6\}$; (3) We insert all tuples in $\pi_{a_2, l, a_4} M$ into $R_7$, with $\{a_2 \rightarrow x_4, l \rightarrow x_5, a_4 \rightarrow x_9\}$; (4) We insert all tuples in $\pi_{a_1, a_3} M$ into $R_9$ with $\{a_1 \rightarrow x_3, a_3 \rightarrow x_6\}$.  After Step (4), we initialize the tensor and can start handling updates on each vector: (5) For every update $(t)$ on $u_3$ $u_4$, we make the same update on $R_8$, $R_{10}$, respectively; (6) For every update $(t)$ on $u_2$, we make the update $(t, \perp)$ on $R_4$; (7) For every update $(t)$ on $u_1$, we make the update $(\perp, t)$ on $R_2$; (8) After updating $i$-th rounds of vectors, we enumerate $Q_1$.  If there exists any query result, the intersection between $M$ and $u_1^{(i)} \times u_2^{(i)} \times u_3^{(i)} \times u_4^{(i)}$ is non-empty.
        

\end{example}

\begin{example}[Maintaining a Star Query]
Consider the dimension-$4$ star query:
    \[
        Q(x_2, x_3, x_4) \gets R_1(x_1, x_2, x_3, x_4), R_2(x_1), R_3(x_2), R_4(x_3).
    \]
Note that $x_4$ is an output attribute in $R_1$ that does not join with any satellite, but the attribute $x_1 \in \V_k(\mathcal{N})$ is a non-output attribute; thus, we can take $x_4$ into consideration for the dimension. 

\paragraph{Database Instance ($R_1$)}
    Let the current database instance $D$ be such that the Generalized H-index is $h=2$.
    The threshold to classify a key as \textbf{Heavy} is $\tau = (h+1)^{d-1} = 3^3 = 27$.
    
    We construct $R_1$ with roughly 80 tuples distributed as follows:
    
    \begin{itemize}[leftmargin=*]
        \item \textbf{Attribute $x_1$ Distribution:} Value \texttt{A} appears in \textbf{30 tuples} (e.g., tuples $t_1 \dots t_{30}$); Value \texttt{B} appears in \textbf{30 tuples} (e.g., tuples $t_{31} \dots t_{60}$); Value \texttt{C} appears in \textbf{10 tuples} (e.g., tuples $t_{61} \dots t_{70}$); Values \texttt{D}-\texttt{M} appear in \textbf{1 tuple} each.  There are 2 values (\texttt{A}, \texttt{B}) with frequency $\ge 27$. Thus, $x_1$ supports $h=2$. Since there are not 3 values with frequency $3^3=27$, it does not reach $h=3$.
        
        \item \textbf{Attribute $x_2$ Distribution:}  Value \texttt{X} appears in \textbf{40 tuples} (overlapping with $x_1$'s \texttt{A} and \texttt{B} tuples); Value \texttt{Y}: appears in \textbf{35 tuples}; Other values appear $\le 5$ times.  Similar to $x_1$, values \texttt{X} and \texttt{Y} exceed the threshold $27$ that are heavy.
        
        \item \textbf{Attribute $x_3, x_4$ Distribution:}  All values are unique (frequency 1).  Therefore, no value has a frequency that exceeds 27, and they are all Light.

    \end{itemize}


    \paragraph{Partitioning $R_1$}
    Given the threshold $\tau = 27$, the algorithm logically partitions $R_1$ into buckets $R_1^{(b_1 b_2 b_3 b_4)}$ with the following heavy/light keys:
    \begin{center}
    \begin{tabular}{|c|c|c|}
        \hline
        \textbf{Attribute} & \textbf{Heavy Keys} & \textbf{Light Keys} \\
        \hline
        $x_1$ & $\{\texttt{A}, \texttt{B}\}$ & $\{\texttt{C}, \texttt{D}, \dots\}$ \\
        $x_2$ & $\{\texttt{X}, \texttt{Y}\}$ & $\{\dots\}$ \\
        $x_3$ & $\emptyset$ & All values \\
        $x_4$ & $\emptyset$ & All values \\
        \hline
    \end{tabular}
    \end{center}
    
    \begin{itemize}[leftmargin=*]
        \item \textbf{Partition $(1,1,0,0)$ [Heavy $x_1$, Heavy $x_2$]:}
        This partition stores the Cartesian product of the dense keys of the data.
        It contains tuples like $(\texttt{A}, \texttt{X}, \dots)$, $(\texttt{A}, \texttt{Y}, \dots)$, $(\texttt{B}, \texttt{X}, \dots)$, and $(\texttt{B}, \texttt{Y}, \dots)$. 

        \item \textbf{Partition $(1,0,0,0)$ [Heavy $x_1$, Light $x_2$]:}
        Contains tuples where $x_1=\texttt{A}, \texttt{B}$ but $x_2\neq\texttt{X}, \texttt{Y}$.
        Recall the $x_2$ key has frequency $\le 5$ except for \texttt{X} and \texttt{Y}.
        
        \item \textbf{Partition $(0,1,0,0)$ [Light $x_1$, Heavy $x_2$]:}
        Contains tuples where $x_1\neq\texttt{A}, \texttt{B}$ but $x_2=\texttt{X}, \texttt{Y}$
        Recall \texttt{C} has frequency $10$ and all keys in \texttt{D}-\texttt{M} have frequency $1$. 
        
        \item \textbf{Partition $(0,0,0,0)$ [All Light]:}
        Contains all remaining tuples, e.g., singletons like $(\texttt{D}, \texttt{Z}, \dots)$.

        \item \textbf{Other partitions:} For other partitions where $x_3$ is heavy or $x_4$ is heavy, they are all empty because none of these keys exist.
    \end{itemize}

    \paragraph{Maintenance Procedure} 
    
    \textbf{Scenario A: Update on a Light Key.} Imagine a tuple $(\texttt{C})$ is inserted into satellite $R_2(x_1)$. The algorithm checks if $\texttt{C} \in K_{x_1}^H$. In this case, it is not.  Therefore, the algorithm must scan the light partitions of $R_1$ (where $b_1=0$) for tuples matching $x_1=\texttt{C}$.  We know from the distribution that $\texttt{C}$ appears exactly $10$ times. Since $10 < \tau=27$, the scan cost is strictly bounded by $O((h+1)^{\dim(Q)-1})$.
    
    \textbf{Scenario B: Update on a Heavy Key.} Imagine a tuple $(\texttt{A})$ is inserted into satellite $R_2(x_1)$. The algorithm checks if $\texttt{A} \in K_{x_1}^H$. In this case, it is \textbf{Heavy}.  Therefore, the algorithm does \emph{not} scan $R_1$ for $\texttt{A}$ (which would cost 30 operations, potentially violating the strict bound if it were higher). Instead, it updates the Heavy View $V_C$.  This view effectively tracks the Cartesian product of heavy keys: $\{\texttt{A}, \texttt{B}\} \times \{\texttt{X}, \texttt{Y}\}$.  It marks $x_1=\texttt{A}$ as alive.  The cost is proportional to the size of the Cartesian product of all heavy keys, which is limited by $h^{\dim(Q)-1} \approx 2^3 = 8$ operations.

    \textbf{Scenario C: Update on $R_1$.}  Imagine the tuple $(a_1, a_2, a_3, a_4)$ is inserted into the center $R_1$.  The algorithm checks if $a_1 \in K_{x_1}^H$, $a_2 \in K_{x_2}^H$, $a_3 \in K_{x_3}^H$, and $a_4 \in K_{x_4}^H$ and obtain the vector $\mathbf{b}$.  After that,  the algorithm inserts $(a_1, a_2, a_3, a_4)$ into $R_1^{(\mathbf{b})}$, for $a_1, \cdots, a_3$, it checks if $a_i \in R_i^{(0)}$ when $a_i$ is light value;  if the tuple passes all these checks, the algorithm projects the tuple to heavy key only and propagates it to $V_C$.  For such updates, the cost is constant.

    \textbf{Scenario D: Key Reclassification.} Imagine the algorithm has deleted $17$ tuples from $R_1$ with $x_1 = \texttt{A}$, making the key $\texttt{A}$ becomes light, while the frequency for \texttt{X} and \texttt{Y} is still higher than $27$ after these deletions.  Then the algorithm deletes the remaining $13$ tuples with $x_1 = \texttt{A}$ from $R_1$, and then all tuples with $x_1 = \texttt{A}$ from $R_2$, the deletion takes $O(1)$ time per tuple in $R_1$, and the deletion in $R_2$ also takes $O(1)$ time due to no propagation can be triggered because $\texttt{A}$ is fully removed from $R_1$.  the algorithm then reinserts all tuples into $R_2$ with the light label, which takes $O(1)$ time because the current $R_1$ contains no $x_1 = \texttt{A}$.  After that, the algorithm finally reinserts all tuples into $R_1$ with the new label, taking $O(1)$ time per tuple.  The total cost of the rebalance step is $O(13)$, and it takes $17$ deletions to perform. By amortizing the cost into these deletions, we can see that the amortized cost is $O(1)$.  

    \textbf{Scenario E: Global Parameter Update.}  Imagine after inserting $53$ tuples with $x_1 = \texttt{D}$ and $44$ tuples with $x_1 = \texttt{C}$ into $R_1$, it makes both $\texttt{C}$ and $\texttt{D}$ become heavy, and key reclassifications are done for these two keys.  We assume these new tuples contain unique $x_2, x_3, x_4$ values, thus won't change the heavy/light classification for other keys.  Now the number of heavy keys for $x_1$ is doubled.  A global parameter update is required.  By doing so, $h_1(D)$ becomes $3$ while $h_2(D), h_3(D), h_4(D)$ remains unchanged.   Thus, we only need to reclassify the heavy keys in $x_1$.  All keys for $x_1$ are light because their frequencies are all smaller than $4^3 = 64$. The algorithm will first delete all heavy keys and reinstall them as light keys.  The total cost for the rebalance can be $O(h(D)^{\dim(Q)})$, however, it is after $h(D)\times h(D)^{\dim{Q}-1}$ updates ($h$ keys are added to heavy, while each key's frequency is change from at most $h(D)^{\dim(Q)-1}$ to at least $2h(D)^{\dim(Q)-1}$.  Therefore, the amortized cost for such updates is $O(1)$.
       
    
\end{example}

\begin{acks}
    This work is supported by the Ministry of Education, Singapore, under its Academic Research Fund Tier 1 (RS32/25), and by the Nanyang Technological University Startup Grant.  The author would like to thank Ce Jin, Haozhe Zhang, and the anonymous reviewers for the insightful discussion and comments.  The author would also like to thank Prof. Ke Yi for his continuous guidance. 
\end{acks}

\clearpage
\bibliographystyle{ACM-Reference-Format}
\bibliography{paper}

@String{Computing = "Computing" }

@String{Computer = "{IEEE} Computer" }

@String{Springer = "Springer-Verlag" }

@ArtifactSoftware{R,
    title = {R: A Language and Environment for Statistical Computing},
    author = {{R Core Team}},
    organization = {R Foundation for Statistical Computing},
    address = {Vienna, Austria},
    year = {2019},
    url = {https://www.R-project.org/},
}

@PREAMBLE{ {\providecommand{\noopsort}[1]{}} }

@string{stoc =  "Proc. ACM Symposium on Theory of Computing"}

@string{focs =  "Proc. IEEE Symposium on Foundations of Computer Science"}

@string{pods =  "Proc. ACM Symposium on Principles of Database Systems"}

@string{graph = "Proc. Graph-Theoretic Concepts in Computer Science"}

@string{vldb =  "Proc. International Conference on Very Large Data Bases"}

@string{icalp= "Proc.  International Colloquium on Automata, Languages,
                  and Programming"}

@string{sigmod = "Proc. ACM SIGMOD International Conference on Management of Data"}

@string{itcs = "Proc. Innovations in Theoretical Computer Science"}

@string{computer = "IEEE Computer"}

@book{abiteboul1995foundations,
	title={Foundations of databases},
	author={Abiteboul, Serge and Hull, Richard and Vianu, Victor},
	year={1995},
	publisher={Addison-Wesley Longman Publishing Co., Inc.}
}

@inproceedings{bagan2007acyclic,
  author={Bagan, Guillaume and Durand, Arnaud and Grandjean, Etienne},
title="On Acyclic Conjunctive Queries and Constant Delay Enumeration",
booktitle="Computer Science Logic",
year="2007",
publisher="Springer Berlin Heidelberg",
address="Berlin, Heidelberg",
pages="208--222"}

@InProceedings{berkholz17:_answer,
author = {Berkholz, Christoph and Keppeler, Jens and Schweikardt, Nicole},
title = {Answering Conjunctive Queries under Updates},
year = {2017},
isbn = {9781450341981},
publisher = {Association for Computing Machinery},
address = {New York, NY, USA},
url = {https://doi.org/10.1145/3034786.3034789},
doi = {10.1145/3034786.3034789},
booktitle = {Proceedings of the 36th ACM SIGMOD-SIGACT-SIGAI Symposium on Principles of Database Systems},
pages = {303–318},
numpages = {16},
location = {Chicago, Illinois, USA},
series = {PODS '17}
}

@article{chirkova2012materialized,
  title={Materialized views},
  author={Chirkova, Rada and Yang, Jun},
  journal={Foundations and Trends{\textregistered} in Databases},
  volume={4},
  number={4},
  pages={295--405},
  year={2012},
  publisher={Now Publishers, Inc.}}

@article{idris2018conjunctive,
  title={Conjunctive queries with inequalities under updates},
  author={Idris, Muhammad and Ugarte, Mart{\'\i}n and Vansummeren, Stijn and Voigt, Hannes and Lehner, Wolfgang},
  journal=vldb,
  volume={11},
  number={7},
  pages={733--745},
  year={2018},
  publisher={VLDB Endowment}}

@inproceedings{lincoln2018tight,
  title={Tight hardness for shortest cycles and paths in sparse graphs},
  author={Lincoln, Andrea and Williams, Virginia Vassilevska and Williams, Ryan},
  booktitle={Proceedings of the Twenty-Ninth Annual ACM-SIAM Symposium on Discrete Algorithms},
  pages={1236--1252},
  year={2018},
  organization={SIAM}
}

@article{kara2020maintaining,
author = {Kara, Ahmet and Ngo, Hung Q. and Nikolic, Milos and Olteanu, Dan and Zhang, Haozhe},
title = {Maintaining Triangle Queries under Updates},
year = {2020},
issue_date = {September 2020},
publisher = {Association for Computing Machinery},
address = {New York, NY, USA},
volume = {45},
number = {3},
issn = {0362-5915},
url = {https://doi.org/10.1145/3396375},
doi = {10.1145/3396375},
journal = {ACM Trans. Database Syst.},
month = {aug},
articleno = {11},
numpages = {46}
}

@inproceedings{nikolic2018incremental,
  title={Incremental view maintenance with triple lock factorization benefits},
  author={Nikolic, Milos and Olteanu, Dan},
  booktitle=sigmod,
  pages={365--380},
  year={2018},
  organization={ACM}
}

@inproceedings{carmeli2019enumeration,
  title={On the Enumeration Complexity of Unions of Conjunctive Queries},
  author={Carmeli, Nofar and Kr{\"o}ll, Markus},
  booktitle={Proceedings of the 38th ACM SIGMOD-SIGACT-SIGAI Symposium on Principles of Database Systems},
  pages={134--148},
  year={2019},
  organization={ACM}
}

@InProceedings{idris17:_dynam,
author = {Idris, Muhammad and Ugarte, Martin and Vansummeren, Stijn},
title = {The Dynamic Yannakakis Algorithm: Compact and Efficient Query Processing Under Updates},
year = {2017},
isbn = {9781450341974},
publisher = {Association for Computing Machinery},
address = {New York, NY, USA},
url = {https://doi.org/10.1145/3035918.3064027},
doi = {10.1145/3035918.3064027},
booktitle = {Proceedings of the 2017 ACM International Conference on Management of Data},
pages = {1259–1274},
numpages = {16},
location = {Chicago, Illinois, USA},
series = {SIGMOD '17}
}

@Article{carbone15:_apach_flink,
  author = 	 {Paris Carbone and Asterios Katsifodimos and Stephan Ewen and Volker Markl and Seif Haridi and Kostas Tzoumas},
  title = 	 {Apache {Flink}: Stream and Batch Processing in a Single Engine},
  journal = 	 {IEEE Data Engineering Bulletin},
  year = 	 2015,
  volume = 	 38,
  number = 	 4,
  pages = 	 {28-38}}

@Book{cormen09:_introd,
  author = 	 {T. H. Cormen and C. E. Leiserson and R. L. Rivest and C. Stein},
  title = 	 {Introduction to Algorithms},
  publisher = 	 {The MIT Press},
  year = 	 2009,
  edition = 	 {3rd}}

@inproceedings{wang2020maintaining,
  title={Maintaining Acyclic Foreign-Key Joins under Updates},
  author={Wang, Qichen and Yi, Ke},
  booktitle={Proceedings of the 2020 ACM SIGMOD International Conference on Management of Data},
  pages={1225--1239},
  year={2020}
}

@inproceedings{kara2020trade,
  title={Trade-offs in static and dynamic evaluation of hierarchical queries},
  author={Kara, Ahmet and Nikolic, Milos and Olteanu, Dan and Zhang, Haozhe},
  booktitle={Proceedings of the 39th ACM SIGMOD-SIGACT-SIGAI Symposium on Principles of Database Systems},
  pages={375--392},
  year={2020}
}

@article{wang2023change,
author = {Wang, Qichen and Hu, Xiao and Dai, Binyang and Yi, Ke},
title = {Change Propagation Without Joins},
year = {2023},
issue_date = {January 2023},
publisher = {VLDB Endowment},
volume = {16},
number = {5},
issn = {2150-8097},
url = {https://doi.org/10.14778/3579075.3579080},
doi = {10.14778/3579075.3579080},
journal = {Proc. VLDB Endow.},
month = {jan},
pages = {1046–1058},
numpages = {13}
}

@article{pods25,
author = {Hu, Xiao and Wang, Qichen},
title = {Towards Update-Dependent Analysis of Query Maintenance},
year = {2025},
issue_date = {May 2025},
publisher = {Association for Computing Machinery},
address = {New York, NY, USA},
volume = {3},
number = {2},
url = {https://doi.org/10.1145/3725254},
doi = {10.1145/3725254},
journal = {Proc. ACM Manag. Data},
month = jun,
articleno = {117},
numpages = {25}
}

@inproceedings{abboud2014popular,
  title={Popular conjectures imply strong lower bounds for dynamic problems},
  author={Abboud, Amir and Williams, Virginia Vassilevska},
  booktitle={2014 IEEE 55th Annual Symposium on Foundations of Computer Science},
  pages={434--443},
  year={2014},
  organization={IEEE}
}

@inproceedings{FindingCycleUpperBound,
author = {Dahlgaard, S\o{}ren and Knudsen, Mathias B\ae{}k Tejs and St\"{o}ckel, Morten},
title = {Finding even cycles faster via capped k-walks},
year = {2017},
isbn = {9781450345286},
publisher = {Association for Computing Machinery},
address = {New York, NY, USA},
url = {https://doi.org/10.1145/3055399.3055459},
doi = {10.1145/3055399.3055459},
booktitle = {Proceedings of the 49th Annual ACM SIGACT Symposium on Theory of Computing},
pages = {112--120},
numpages = {9},
location = {Montreal, Canada},
series = {STOC 2017}
}

@InProceedings{lincoln_et_al:LIPIcs.ITCS.2020.11,
  author =	{Lincoln, Andrea and Vyas, Nikhil},
  title =	{{Algorithms and Lower Bounds for Cycles and Walks: Small Space and Sparse Graphs}},
  booktitle =	{11th Innovations in Theoretical Computer Science Conference (ITCS 2020)},
  pages =	{11:1--11:17},
  series =	{Leibniz International Proceedings in Informatics (LIPIcs)},
  ISBN =	{978-3-95977-134-4},
  ISSN =	{1868-8969},
  year =	{2020},
  volume =	{151},
  editor =	{Vidick, Thomas},
  publisher =	{Schloss Dagstuhl -- Leibniz-Zentrum f{\"u}r Informatik},
  address =	{Dagstuhl, Germany},
  URL =		{https://drops.dagstuhl.de/entities/document/10.4230/LIPIcs.ITCS.2020.11},
  URN =		{urn:nbn:de:0030-drops-116969},
  doi =		{10.4230/LIPIcs.ITCS.2020.11}
}

@article{colorcoding,
author = {Alon, Noga and Yuster, Raphael and Zwick, Uri},
title = {Color-coding},
year = {1995},
issue_date = {July 1995},
publisher = {Association for Computing Machinery},
address = {New York, NY, USA},
volume = {42},
number = {4},
issn = {0004-5411},
url = {https://doi.org/10.1145/210332.210337},
doi = {10.1145/210332.210337},
journal = {J. ACM},
month = jul,
pages = {844–856},
numpages = {13}
}

@inproceedings{generalizedOMv,
author = {Jin, Ce and Xu, Yinzhan},
title = {Tight dynamic problem lower bounds from generalized BMM and OMv},
year = {2022},
isbn = {9781450392648},
publisher = {Association for Computing Machinery},
address = {New York, NY, USA},
url = {https://doi.org/10.1145/3519935.3520036},
doi = {10.1145/3519935.3520036},
booktitle = {Proceedings of the 54th Annual ACM SIGACT Symposium on Theory of Computing},
pages = {1515–1528},
numpages = {14},
location = {Rome, Italy},
series = {STOC 2022}
}

@techreport{grahamGYO,
    author = {Marc H. Graham},
    title = {On the universal relation},
    institution = {University of Toronto},
    year = 1979
}

@inproceedings{DBLP:conf/compsac/YuO79,
  author       = {C. T. Yu and
                  M. Z. \"{O}zsoyo\u{g}lu},
  title        = {An algorithm for tree-query membership of a distributed query},
  booktitle    = {The {IEEE} Computer Society's Third International Computer Software
                  and Applications Conference, {COMPSAC} 1979, 6-8 November, 1979, Chicago,
                  Illinois, {USA}},
  pages        = {306--312},
  publisher    = {{IEEE}},
  year         = {1979},
  url          = {https://doi.org/10.1109/CMPSAC.1979.762509},
  doi          = {10.1109/CMPSAC.1979.762509},
  bibsource    = {dblp computer science bibliography, https://dblp.org}
}

@inproceedings{DBLP:conf/focs/AbboudBW15a,
  author       = {Amir Abboud and
                  Arturs Backurs and
                  Virginia {Vassilevska Williams}},
  editor       = {Venkatesan Guruswami},
  title        = {If the Current Clique Algorithms are Optimal, So is Valiant's Parser},
  booktitle    = {{IEEE} 56th Annual Symposium on Foundations of Computer Science, {FOCS}
                  2015, Berkeley, CA, USA, 17-20 October, 2015},
  pages        = {98--117},
  publisher    = {{IEEE} Computer Society},
  year         = {2015},
  url          = {https://doi.org/10.1109/FOCS.2015.16},
  doi          = {10.1109/FOCS.2015.16},
  bibsource    = {dblp computer science bibliography, https://dblp.org}
}

@article{DBLP:journals/algorithmica/AlonYZ97,
  author       = {Noga Alon and
                  Raphael Yuster and
                  Uri Zwick},
  title        = {Finding and Counting Given Length Cycles},
  journal      = {Algorithmica},
  volume       = {17},
  number       = {3},
  pages        = {209--223},
  year         = {1997},
  url          = {https://doi.org/10.1007/BF02523189},
  doi          = {10.1007/BF02523189},
  bibsource    = {dblp computer science bibliography, https://dblp.org}
}

@inproceedings{DBLP:conf/focs/NaorSS95,
  author       = {Moni Naor and
                  Leonard J. Schulman and
                  Aravind Srinivasan},
  title        = {Splitters and Near-Optimal Derandomization},
  booktitle    = {36th Annual Symposium on Foundations of Computer Science, {FOCS} 1995,
                  Milwaukee, Wisconsin, USA, 23-25 October 1995},
  pages        = {182--191},
  publisher    = {{IEEE} Computer Society},
  year         = {1995},
  url          = {https://doi.org/10.1109/SFCS.1995.492475},
  doi          = {10.1109/SFCS.1995.492475},
  bibsource    = {dblp computer science bibliography, https://dblp.org}
}

@article{comparison,
author = {Wang, Qichen and Yi, Ke},
title = {Conjunctive Queries with Comparisons},
year = {2025},
publisher = {Association for Computing Machinery},
address = {New York, NY, USA},
issn = {0362-5915},
url = {https://doi.org/10.1145/3769424},
doi = {10.1145/3769424},
journal = {ACM Trans. Database Syst.},
month = sep
}

@article{difference,
author = {Hu, Xiao and Wang, Qichen},
title = {Computing the Difference of Conjunctive Queries Efficiently},
year = {2023},
issue_date = {June 2023},
publisher = {Association for Computing Machinery},
address = {New York, NY, USA},
volume = {1},
number = {2},
url = {https://doi.org/10.1145/3589298},
doi = {10.1145/3589298},
journal = {Proc. ACM Manag. Data},
month = jun,
articleno = {153},
numpages = {26}
}

@inproceedings{conjunctiveLB,
author = {Mengel, Stefan},
title = {Lower Bounds for Conjunctive Query Evaluation},
year = {2025},
isbn = {9798400715655},
publisher = {Association for Computing Machinery},
address = {New York, NY, USA},
url = {https://doi.org/10.1145/3722234.3725824},
doi = {10.1145/3722234.3725824},
booktitle = {Companion of the 44th Symposium on Principles of Database Systems},
pages = {5},
numpages = {1},
location = {Berlin, Germany},
series = {PODS '25}
}

@InProceedings{10.1007/978-3-642-03367-4_25,
author="Eppstein, David
and Spiro, Emma S.",
editor="Dehne, Frank
and Gavrilova, Marina
and Sack, J{\"o}rg-R{\"u}diger
and T{\'o}th , Csaba D.",
title="The h-Index of a Graph and Its Application to Dynamic Subgraph Statistics",
booktitle="Algorithms and Data Structures",
year="2009",
publisher="Springer Berlin Heidelberg",
address="Berlin, Heidelberg",
pages="278--289",
isbn="978-3-642-03367-4"
}

@article{EPPSTEIN201244,
title = {Extended dynamic subgraph statistics using h-index parameterized data structures},
journal = {Theoretical Computer Science},
volume = {447},
pages = {44-52},
year = {2012},
note = {Combinational Algorithms and Applications (COCOA 2010)},
issn = {0304-3975},
doi = {https://doi.org/10.1016/j.tcs.2011.11.034},
url = {https://www.sciencedirect.com/science/article/pii/S0304397511009534},
author = {David Eppstein and Michael T. Goodrich and Darren Strash and Lowell Trott}
}

@article{PintoR24,
	author       = {Dante Pinto and
	Cristian Riveros},
	title        = {Complex event recognition meets hierarchical conjunctive queries},
	journal      = {CoRR},
	volume       = {abs/2408.01652},
	year         = {2024},
	url          = {https://doi.org/10.48550/arXiv.2408.01652},
	doi          = {10.48550/ARXIV.2408.01652},
	eprinttype    = {arXiv},
	eprint       = {2408.01652},
	bibsource    = {dblp computer science bibliography, https://dblp.org}
}

@article{10.1145/3695837,
author = {Abo Khamis, Mahmoud and Kara, Ahmet and Olteanu, Dan and Suciu, Dan},
title = {Insert-Only versus Insert-Delete in Dynamic Query Evaluation},
year = {2024},
issue_date = {November 2024},
publisher = {Association for Computing Machinery},
address = {New York, NY, USA},
volume = {2},
number = {5},
url = {https://doi.org/10.1145/3695837},
doi = {10.1145/3695837},
journal = {Proc. ACM Manag. Data},
month = nov,
articleno = {219},
numpages = {26}
}

@inproceedings{fan2023fine,
  title={The Fine-Grained Complexity of Boolean Conjunctive Queries and Sum-Product Problems},
  author={Fan, Austen Z and Koutris, Paraschos and Zhao, Hangdong},
  booktitle={50th International Colloquium on Automata, Languages, and Programming (ICALP 2023)},
  pages={127--1},
  year={2023},
  organization={Schloss Dagstuhl--Leibniz-Zentrum f{\"u}r Informatik}
}

@inproceedings{10.1145/3584372.3588666,
author = {Deng, Shiyuan and Lu, Shangqi and Tao, Yufei},
title = {On Join Sampling and the Hardness of Combinatorial Output-Sensitive Join Algorithms},
year = {2023},
isbn = {9798400701276},
publisher = {Association for Computing Machinery},
address = {New York, NY, USA},
url = {https://doi.org/10.1145/3584372.3588666},
doi = {10.1145/3584372.3588666},
booktitle = {Proceedings of the 42nd ACM SIGMOD-SIGACT-SIGAI Symposium on Principles of Database Systems},
pages = {99–111},
numpages = {13},
location = {Seattle, WA, USA},
series = {PODS '23}
}

@inproceedings{abboud2024faster,
  title={Faster combinatorial k-clique algorithms},
  author={Abboud, Amir and Fischer, Nick and Shechter, Yarin},
  booktitle={Latin American Symposium on Theoretical Informatics},
  pages={193--206},
  year={2024},
  organization={Springer}
}

@article{HAAS2006360,
title = {Chordless paths through three vertices},
journal = {Theoretical Computer Science},
volume = {351},
number = {3},
pages = {360-371},
year = {2006},
note = {Parameterized and Exact Computation},
issn = {0304-3975},
doi = {https://doi.org/10.1016/j.tcs.2005.10.021},
url = {https://www.sciencedirect.com/science/article/pii/S0304397505006304},
author = {Robert Haas and Michael Hoffmann}
}

@inproceedings{10.1145/3584372.3588667,
author = {Carmeli, Nofar and Segoufin, Luc},
title = {Conjunctive Queries With Self-Joins, Towards a Fine-Grained Enumeration Complexity Analysis},
year = {2023},
isbn = {9798400701276},
publisher = {Association for Computing Machinery},
address = {New York, NY, USA},
url = {https://doi.org/10.1145/3584372.3588667},
doi = {10.1145/3584372.3588667},
booktitle = {Proceedings of the 42nd ACM SIGMOD-SIGACT-SIGAI Symposium on Principles of Database Systems},
pages = {277–289},
numpages = {13},
location = {Seattle, WA, USA},
series = {PODS '23}
}

\end{document}